\documentclass[11pt, letterpaper]{article}
\usepackage{amssymb,amsmath,latexsym,amscd,amsfonts, bbm}  
\usepackage{graphicx}  
\usepackage[usenames]{color}

\newtheorem{theorem}{Theorem}[section]

\newtheorem{deff}[theorem]{Definition}  
  
\newtheorem{claim}[theorem]{Claim}  
\newtheorem{lem}[theorem]{Lemma}  
\newtheorem{conj}[theorem]{Conjecture}

\newcommand{\qedsymb}{\hfill{\rule{2mm}{2mm}}}  
\newenvironment{proof}[1][]{\begin{trivlist}  
\item[\hspace{\labelsep}{\bf\noindent Proof#1:\/}] 
}{\qedsymb\end{trivlist}}

\newcommand{\ignore}[1]{}

\newcommand{\BPP}{\mathsf{BPP}}
\newcommand{\RP}{\mathsf{RP}}
\newcommand{\QMA}{\mathsf{QMA}}

\newcommand{\NP}{\mathsf{NP}}

\newcommand{\PCP}{\mathsf{PCP}}
\newcommand{\QPCP}{\mathsf{QPCP}}

\newcommand{\B}[1]{\mathbf{#1}}

\newcommand{\Be}{\B{e}}

\newcommand{\norm}[1]{{\| #1 \|}}  
  
\newcommand{\ket}[1]{{ |{#1} \rangle }}  
\newcommand{\bra}[1]{{ \langle {#1} | }}

\newcommand{\orderof}[1]{\mathcal{O}(#1)} 
\newcommand{\poly}{\mathrm{poly}} 
\newcommand{\EqDef}{\stackrel{\mathrm{def}}{=}}

\newcommand{\Pl}[1]{\Pi^{\scriptscriptstyle(#1)}_{\scriptscriptstyle \le \ell}}  
\newcommand{\Pgt}[1]{\Pi^{\scriptscriptstyle(#1)}_{\scriptscriptstyle > \ell}}  
\newcommand{\Ppr}[2]{\Delta^{\scriptscriptstyle(#1)}_{#2}}
\newcommand{\Prs}[2]{R^{\scriptscriptstyle(#1)}_{\scriptscriptstyle \le #2}}
\newcommand{\UNSAT}{\mathrm{UNSAT}}
\newcommand{\QUNSAT}{\mathrm{QUNSAT}}
\newcommand{\SAT}{\mathrm{SAT}}
\newcommand{\QSAT}{\mathrm{QSAT}}
\newcommand{\kProj}{\text{$k$-QSAT}}
\newcommand{\Av}{\mathbbm{E}}
\newcommand{\Qbit}{\mathbbm{B}}

\newcommand{\Eq}[1]{Eq.~(\ref{#1})}
\newcommand{\Fig}[1]{Fig.~\ref{#1}}

\newcommand{\Lem}[1]{Lemma~\ref{#1}}
\newcommand{\Cl}[1]{Claim~\ref{#1}}

\newcommand{\Sec}[1]{Sec.~\ref{#1}}
\newcommand{\Ref}[1]{Ref.~\cite{#1}}

\newcommand{\App}[1]{Appendix~\ref{#1}}

\newcommand{\Id}{\mathbbm{1}}

\begin{document}

\title{The Detectability Lemma and Quantum Gap Amplification} 
\author{Dorit Aharonov\thanks{School of Computer Science and Engineering, 
 The Hebrew University, Jerusalem, Israel.},\ \  
   Itai Arad\thanks{Department of Electrical Engineering and Computer Sciences, 
   University of California at Berkeley, CA.},\ \  
   Zeph Landau\thanks{Department of Electrical Engineering and Computer Sciences,
 University of California at Berkeley, CA.}\ 
 \ and Umesh Vazirani \thanks{Department of Electrical Engineering and Computer Sciences,
 University of California at Berkeley, CA.}
} 
\maketitle 
 
\noindent 
\begin{abstract} 
  The quantum analogue of a constraint satisfaction problem is a sum
  of local Hamiltonians - each (term of the) Hamiltonian specifies a
  local constraint whose violation contributes to the energy of the
  given quantum state. Formalizing the intuitive connection between
  the ground (minimal) energy of the Hamiltonian and the
  \emph{minimum} number of violated constraints is problematic,
  since the number of constraints being violated is not well defined
  when the terms in the Hamiltonian do not commute. The
  detectability lemma proved in this paper provides precisely such a
  quantitative connection. We apply the lemma to derive a quantum
  analogue of a basic primitive in classical computational
  complexity: amplification of probabilities by random walks on
  expander graphs. We call it the quantum gap amplification lemma. 
  It holds under the restriction that the interaction graph of the
  local Hamiltonian is an expander. Our proofs are based on a novel
  structure imposed on the Hilbert space that we call the $XY$
  decomposition, which enables a reduction from the quantum
  non-commuting case to the commuting case (where many classical
  arguments go through).

  The results may have several interesting implications.  First,
  proving a quantum analogue to the PCP theorem is one of the most
  important challenges in quantum complexity theory. Our quantum gap
  amplification lemma may be viewed as the quantum analogue of the
  first of the three main steps in Dinur's PCP proof
  \cite{ref:Din07}. Quantum gap amplification may also be related to
  spectral gap amplification, and in particular, to fault tolerance
  of adiabatic computation, a model which has attracted much
  attention but for which no fault tolerance theory was derived yet.
  Finally, the detectability lemma, and the $XY$ decomposition
  provide a handle on the structure of local Hamiltonians and their
  ground states. This may prove useful in the study of those
  important objects, in particular in the fast growing area of
  ``quantum Hamiltonian complexity'' connecting quantum complexity
  to condensed matter physics. 

\end{abstract}

\section{Introduction}

There is a close analogy between two fundamental notions from
computational complexity theory and quantum physics: constraint
satisfaction problems and the ground energy of local Hamiltonians.
Each term in the local Hamiltonian specifies a local constraint
whose violation contributes to the energy of the given quantum
state. Hence the energy of the quantum state corresponds intuitively
to the number of violated quantum constraints. A canonical example
of this is the correspondence between the classical Cook-Levin
theorem and its quantum analogue proved by Kitaev \cite{ref:Kit02}.
Kitaev showed that estimating the ground energy of a local
Hamiltonian to within inverse polynomial accuracy (the quantum
analogue of determining the minimal number of violated constraints)
is complete for the quantum analog of $\NP$, namely $\QMA$.

But how accurate is this intuitive correspondence between the energy
of a state and the number of violated quantum constraints? The main
issue is that in the quantum case, the terms of the Hamiltonian do
not commute in general. This means that it is not even meaningful to
ask: how many constraints are violated by a given state? Or in
keeping with the probabilistic nature of quantum physics: what is
the probability that the given state violates at least $k$
constraints?

Our first main result in this paper is the \emph{quantum
detectability lemma}, which provides a way of making sense of and
answering these questions.  The lemma applies to any local
Hamiltonian subject to the mild restrictions that every particle
(qubit) in the local Hamiltonian participates in a bounded number of
constraints, and each term in the Hamiltonian is chosen from a
finite set of possibilities.

To state the detectability lemma, consider partitioning the terms in
the Hamiltonian into $g$ sets, which we call layers, so that in each
layer all terms are mutually commuting. Under the restrictions on
the Hamiltonian, it is possible to choose the number of layers $g$
to be a constant.  Notice that in every one of the $g$ layers, it is
meaningful to ask how many constraints are violated; every quantum
state $\ket{\psi}$ induces a probability distribution on how many
constraints are violated for a given layer.  The detectability lemma
states that if the ground energy of the system is finite, then the
probability that one or more constraints are violated in at least
one of the layers is also finite.  This is the probability of
detecting one or more violations when measuring the constraints in
that layer -- hence the name \emph{detectability lemma}. In its most
general form, the detectability lemma also ensures that for systems
with high ground energy, there exists a layer in which the
probability for more than $\ell$ violations is finite. Here $\ell$
is some integer that has to be smaller than some normalized version
of the ground energy.

To understand the subtleties  of the lemma, consider
a system with a ground state $\epsilon_0>0$, in which every layer
consists of $m$ constraints. Now consider the distribution induced
by some state {$\ket{\psi}$} on the different \emph{sectors} of the
layers (by sector we mean the subspace corresponding to a certain
number of violated constraints in that layer). One can imagine that
the induced distribution in every layer has a tiny $\epsilon_0/m$
weight in the $m$ violations part, and the rest is concentrated in
the $0$ violations part with no weight in the intermediate part
(with $1,2, \ldots, m-1$ violations).  Such a setup would certainly
comply with the ground energy condition, but constraint violation
would not be detectable.  When the different layers commute, it is
easy to see that this scenario cannot happen, since it would imply
a common ground state for all layers which contradicts $\epsilon_0
>0$. However, when the layers do not commute, the relationship
between the distributions becomes non-trivial. The $\ell=0$
detectability lemma shows that even in the non-commuting case, the
above  scenario cannot happen. It shows that there is at least one
layer in which the total weight on one or more violations is larger
than some constant that is linear in $\epsilon_0$ (for small
$\epsilon_0$) and \emph{is independent of the system size $m$}.

The heart of the proof is a certain decomposition of the Hilbert
called the $XY$-decomposition, which is interesting in its own
right. It captures a structural relationship between the ground
spaces of the different layers of the Hamiltonian. The decomposition
first partitions the Hilbert space into a tensor product of local
spaces (defined by objects which we call \emph{pyramids}) and then
further decomposes each of these local spaces into commuting and
non-commuting parts with respect to the Hamiltonian. Roughly
speaking, the commuting parts are dealt with by classical means.  In
the non-commuting parts we identify an important parameter $0<
\theta<1$ of the system (characterized by the finite family of
constraints allowed in the Hamiltonian) and find a way to point at
an exponential decay of the states in terms of that parameter. This
exponential decay allows for local analysis of the actions of the
individual terms of the Hamiltonian. 

Classical gap amplification, first proved in the context of saving
random bits in $\RP$ and $\BPP$ amplification \cite{ref:Ajt87,
ref:Imp89}, is a basic primitive in complexity theory.  The idea is
that if one is interested in amplifying the probability of hitting a
given subset of the nodes (or edges) in a graph, then if the graph
is an expander, a random walk would do almost as well as picking the
nodes (or edges) independently. More generally a constraint
satisfaction problem is represented by a hypergraph, with each
hyperedge corresponding to a constraint. To amplify the gap between
the acceptance and rejection probability, one considers the
``$t$-step walk'' on the hypergraph. Now, if the hypergraph is
expanding, one can show that the gap gets amplified by a factor of
$\Omega(t)$. This idea has since found many other important
implications, for example, in Dinur's proof of the PCP theorem
\cite{ref:Din07}. 

In this paper we prove a quantum analogue of the classical
amplification lemma: the hyper-constraints are also generated from
$t$ terms in the original Hamiltonian which form a walk in the
interaction graph. Each new hyper-constraint is the projection on
the intersection of the $t$ constraints on the walk. We show that if
the original interaction graph was an expander, the average ground
energy per term of the new Hamiltonian (consisting of the
hyper-constraints) is $\Omega(t)$ times the average ground energy
per term of the original Hamiltonian, thus establishing a bound
similar to the classical lemma. The proof relies critically on the
quantum detectability lemma, along with the classical analysis of
walks on expander graphs. 

The idea of the proof is that the overall amplification is
lower-bounded by the amplification of a single layer. But as the
constraints of a layer commute, we can treat them classically and
apply the classical amplification lemma to the distribution of
violations at that layer. The amplification of a layer therefore
depends on its distribution. This is exactly where the detectability
lemma is needed, as it ensures us that there is at least one layer
with a distribution that allows for substantial amplification.

{~}

\noindent \textbf{Discussions and Possible Implications:} 
The results in this paper are related to several important open
problems in quantum computation complexity. First, the study of the
computational complexity of local Hamiltonians has blossomed over
the last few years, and touches upon efficient simulation of quantum
systems and theoretical condensed matter physics. The techniques
developed in this paper, the $XY$-decomposition and the quantum
detectability lemma, can be expected to contribute to our
understanding of this new area.

Second, the PCP theorem is arguably the most important development
in computational complexity theory over the last two decades. Is
there a quantum analogue? One natural formulation is the following:
suppose we are given a local Hamiltonian on $n$ qubits with the
promise that the ground energy is either $0$ or at least $1/p(n)$
for some polynomial $p(n)$. Is there a way to map this to a new
local Hamiltonian such that the ground energy is either $0$ or
$\Omega(n)$?. Proving such a quantum PCP theorem is a major challenge
in quantum complexity theory; it would have implications for our
understanding of inapproximability results of quantum complexity
problems, quantum fault tolerance, the understanding of entanglement
and notions such as no-cloning, as well as on the basic notion of
energy gap amplification in condensed matter physics 
(see Section \ref{sec:qpcp} for more precise
definition and discussion). 

Our quantum gap amplification lemma can be viewed as a very weak
form of the above statement of quantum PCP. 
The problem of course is that checking the new
$t$-walk constraints requires $t$ queries, which is too large even
if we wish to check a {\it single} constraint. Dinur's proof of the
classical PCP theorem combines this kind of gap amplification with
two other steps - degree reduction and assignment testing. 
In this sense quantum gap amplification is a
possible first step towards emulating the outline of Dinur's proof
in the quantum setting.

Gap amplification is tightly connected (though not the same!) to
spectral gap amplification, a notion of interest in adiabatic
quantum computation (and in condensed matter physics in general).  
Adiabatic computation is a model of quantum computation which is 
equivalent in power to the standard one, and has attracted
considerable attention (\cite{ref:Far00, ref:Dam01, ref:Dam01b,
ref:Rei04, ref:Rol02, ref:Far02, ref:Aha03, ref:Aha04} and more). 
In adiabatic computation, the system evolves under a Hamiltonian
with a non-negligible spectral gap between the ground state and the
next excited state. Physical intuition suggests that such a model
might be inherently robust to thermal noise \cite{ref:Chi01}.
Despite work on the subject
\cite{ref:Lbe05,ref:Rol05,ref:Lid05,ref:Lid08}, including the
development of quantum error correcting codes tailored for adiabatic
evolution \cite{ref:Jor06}, an analogue to the threshold result of
the standard model \cite{ref:Aha97, ref:Kni98, ref:Kit03} is still
missing. Can the spectral gap in adiabatic computation be amplified
to a constant, to provide fault tolerance? This is probably
impossible when the system of qubits is arranged on a line, since
even though such a system can be adiabatically universal when the
gap is inverse polynomial \cite{ref:Aha07b}, Hastings has showed
\cite{ref:Has07, ref:Osb07} that adiabatic evolution in one
dimension with constant spectral gap can be simulated efficiently
classically. However, it may very well be true that amplification of
the spectral gap in some well defined sense is possible, if the
underlying geometry is that of an expander.  It seems likely that a
proof of a quantum PCP theorem would pave the way to such a result,
though we should caution that there is no proof showing such an
implication.

{~}

\noindent\textbf{Open Problems} 
In this paper we handle the restricted case in which the local terms
in the Hamiltonian are projections. We leave the general case for
future work.

As a benchmark open problem, we pose the following question: prove
an exponential size quantum PCP, in analogy with the first classical
PCP results \cite{ref:Aro08}. This already seems to require some
non-trivial work in quantum information theory. We note that it is
possible to prove a quantum PCP theorem where the proof is of doubly
exponential size; this seems to show that the no-cloning theorem, which 
some believe to be an obstacle against quantum PCP, might in fact be possible 
to bypass. 

Proving a quantum analogue of Dinur's degree
reduction, which allows reducing the degree of the graph of
interactions in the Hamiltonian, is another major open problem, 
which is related to the above problem. 

Another related problem is to improve the parameters in current
perturbation gadgets \cite{ref:Kem06, ref:Oli05, ref:Bra08,
ref:Jor08} significantly; Perturbation gadjets 
are objects that allow decomposing 
Hamiltonians acting on some number of qubits, 
to sums of terms acting on smaller sets while maintaining 
some properties of the eigenvalues of the original Hamiltonian 
(within some small error). These are clearly tightly related to 
the notion of degree reduction, however the parameters in current 
perturbation gadgets are not good enough for the purposes of degree 
reduction, since the harm the gap too much.

\section{Background - Local Hamiltonians and Local Projections}

A $k$-local Hamiltonian $H$ on $n$ qubits is an operator
$H:\Qbit^{\otimes n}\to \Qbit^{\otimes n}$ that can be written as a
sum $H=\sum_{i=1}^M H_i$, where $M=poly(n)$ and every $H_i$ is a
Hermitian operator acting on at most $k$ qubits. In this paper we
restrict attention to the case where the $H_i$ operators are
projections.  We will usually denote them by $Q_i$: 

\begin{align}
 H =
\sum_{i=1}^M Q_i. 
\end{align} 

 Another assumption that we use is that every
projection intersects with only a finite number of other
projections. Together with the $k$-locality, this implies that the
projections can be partitioned into a constant number, denoted $g$,
of subsets (which we call {\it layers}) such that the projections in
each layer are non-intersecting and thus commuting. We denote a
system satisfying the above restriction by a $\kProj$ system. 

For a state $\ket{\psi}$, $\bra{\psi}H\ket{\psi} = \sum_i \bra{\psi}
H_i\ket{\psi}$ is called the {\it energy} of the state. In our case,
this can be any number between $0$ and $M$. The minimal energy of
the system is the lowest eigenvalue of the Hamiltonian, and is
denoted by $\epsilon_0$.

Deciding whether $\epsilon_0$ is above some threshold $a$ or below a
threshold $b$ with $a-b>1/poly(n)$ is known as the $k$-local
Hamiltonian problem (which is complete for Quantum NP). It can be
viewed as the quantum analog of the $k-\SAT$ problem.  We often
refer to the projections $Q_i$ as {\it constraints}, and when
$\bra{\psi}Q_i\ket{\psi} > 0$ we say that $Q_i$ is violated (with
respect to the state $\ket{\psi}$). Similarly, when
$\bra{\psi}Q_i\ket{\psi} =0$, we say that $Q_i$ is satisfied, or
that $\ket{\psi}$ is in the accepting subspace of $Q_i$.

\section{The $XY$ decomposition}
\label{sec:XY}

We consider a $\kProj$ system over $n$ qubits with $M$ constraints
that can be arranged in $g$ layers. Let us start by describing at a
high level the $XY$-decomposition and how it is applied:

We start with a decomposition of the Hilbert space into a tensor
product of local spaces, and restrict our attention to only those
terms of the Hamiltonian that act non-trivially on exactly one of
these spaces. The way the actual decomposition is carried out
depends upon an ordering of the layers, and is described in the
pyramid construction below. Each of these local spaces can now be
further decomposed into a direct sum of subspaces according to
whether all the terms of the Hamiltonian acting on this subspace
commute ($X$) or not ($Y$).  This defines a natural
$XY$-decomposition of any state. The main point of the
$XY$-decomposition is that it allows us to capture some structural
relationship between the ground spaces of the different layers of
the Hamiltonian. The starting point for this is the observation that
in each $Y$ subspace there is some finite angle between the ground
spaces of the different layers. Now stepping back, we can decompose
the tensor product of all the local spaces into subspaces according
to the number of $Y$ components. The actions of the Hamiltonian on
the local spaces collectively ensure that if we start from an
arbitrary state and successively project it onto the ground spaces
of the different layers, then the weight of the resulting state in
each subspace decays exponentially in the number of $Y$ components.
This is a key property used in the proof of the detectability lemma.

\subsection{Pyramids and Pyramid projections}

We partition the Hilbert space of the $\kProj$ system into a product
of subspaces by defining the notion of a ``pyramid'', a special
``connected'' subset of the constraints, as follows.  First, we
arbitrarily order the layers from $1$ to $g$.  A pyramid is created
by picking its apex - a constraint in the first layer - and for each
successive layer picking all constraints that intersect with the set
of constraints picked in previous layers. We denote the Hilbert
space of the qubits which participate in the pyramid by $H_{pyr}$.
We now consider any maximal set of disjoint pyramids as illustrated
in \Fig{fig:pyramids}. Clearly the entire Hilbert space can be
written as a tensor product of the pyramid spaces $H_{pyr}$, and
constraints from different pyramids commute.

\begin{figure}
  \center \includegraphics[scale=0.7]{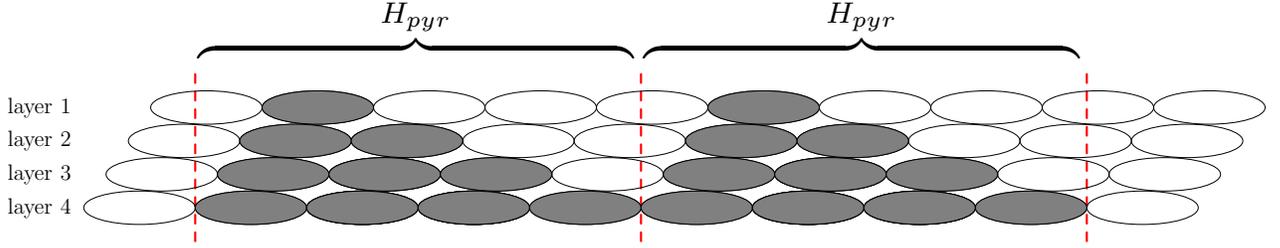} 
  \caption{The pyramids
  \label{fig:pyramids}}
\end{figure}

In the next step, we decompose the Hilbert space $H_{pyr}$ of the
first pyramid into a direct sum of subspaces $\{X_j\}$ and a a
subspace $Y$: every space $X_j$ is made of vectors which are
simultaneous eigenvectors of \emph{all} projections in the pyramid.
Moreover, in every such $X_j$, each projection is allowed to take
only one value - $0$ or $1$. Then $Y$ is defined to be the residual
subspace, i.e., the subspace that is orthogonal to all the $X_j$
subspaces inside $H_{pyr}$.  Clearly, all these spaces are
orthogonal to each other. We refer to $Y$ as the ``non-commuting''
part of the Hilbert space $H_{pyr}$; all other subspaces correspond
to the ``commuting parts''.  Of course, this decomposition can be
done for every one of the pyramids. 

We denote a sector of the $XY$ decomposition by a string $\nu$.
$\nu$ specifies either an $X_i$ space or a $Y$ space at each
location, and we define
\begin{align}
  |\nu| \EqDef \text{No. of $Y$ sites in $\nu$} \ .
\end{align}
We also define $P_\nu$ to be the projection into the tensor product
of these spaces. Note that $P_\nu$ is by itself a product of all the
corresponding $P_{X_i}, P_Y$ projections. Every state in $H$, $\ket{\psi}$,
can therefore be written as
\begin{align}
\label{eq:sectors}
  \ket{\psi} = \sum_\nu P_\nu\ket{\psi} \EqDef \sum_\nu \lambda_\nu
  \ket{\psi_\nu} \ . 
\end{align}
This is the $XY$ decomposition. 

We will in fact eventually use a finite number of these $XY$
decompositions. It is easy to see that there exists a constant
$f(k,g)$ (independent of $n$) of XY decompositions such that every
constraint in the top layer appears in one pyramid top in one of the
$XY$ decompositions, and so all top constraints are ``covered'' by
one of the decompositions.  But for most of the remainder of the
paper, we fix one $XY$ decomposition and stick to it. 

\subsection{Commutation relations between projections inside the 
pyramids} 
For a fixed pyramid, we denote the
operators which  act on $H_{pyr}$ and project of 
on the subspaces $\{X_j\}_j$ and $Y$ 
by $\{P_{X_j}\}_j$, $P_Y$ respectively. It is
easy to verify the following properties:
\begin{itemize}
  \item The projections form a valid decomposition of $H_{pyr}$: 
    \begin{align}
      [P_{X_i},P_{X_j}] &= [P_{X_i}, P_Y] = 0 \ , \\
       P_Y + \sum_j P_{X_j}  &= \Id 
    \end{align}
    
  \item Those projections commute with 
the constraints in the pyramid: 
    \begin{align} \label{eq:commutation1}
      [Q, P_{X_j}] = [Q, P_Y] = 0 \ .
    \end{align}
  
  \item For every two constraints $Q_1, Q_2$ in
   the pyramid  and every subspace $X_j$
    \begin{align}
      P_{X_j} [Q_1, Q_2] P_{X_j} = 0 \ .
    \end{align}
\end{itemize}

\subsection{The parameter $\theta$}
Next, we define the
parameter $0<\theta<1$, which plays a crucial role in the paper. 

\begin{deff}[The parameter $\theta$]
\label{def:theta} 

  Fix a pyramid.  Consider the product $Q_0\cdot Q_1\cdot\ldots\cdot
  Q_N$ where every $Q_i$ is either a projection from the pyramid or
  its complement, and every pyramid projection (or its complement)
  appears exactly once.  $Y$ does not contain any common eigenvector
  of \emph{all} those projections. Hence there exists a constant
  $0<\theta<1$ such that for any possible pyramid in the system, and
  any order in which the constraints are chosen to appear in the
  product, 
  \begin{align}
  \label{eq:theta}
    \norm{P_Y \cdot Q_0\cdot\ldots\cdot Q_N \cdot P_Y} \le \theta 
      \ .
  \end{align}
  $\theta$ is a constant that depends only on the family of
  constraints and on the constant $g$ (which determines the maximal
  number of constraints in a pyramid). 
\end{deff}

\subsection{The $\Pl{i}$ projections}  
\label{sec:Pi}

The $XY$ decomposition is useful for analyzing the action of the
$\Pl{i}$ projections, which play a central role in the exponential
decay and the detectability lemma. The $\Pl{i}$ projection projects
to the subspace of $\ell$ or less violations of the constraints in
the $i$'th layer. For example, for $\ell=0$, $\Pl{i}$ projects into
the accepting space of the $i$'th layer.

A central observation is that we can present $\Pl{i}$
according to the pyramids structure. Using the fact that the
constraints in the $i$'th layer all commute with each other and
defined on non-intersecting qubits, we may write it in terms of
violations inside the pyramids and violations outside the pyramids.
Specifically, we write it as the following sum of $\ell+1$ terms:
\begin{align}
\label{eq:inside-out}
  \Pl{i} = \sum_{j=0}^\ell \Ppr{i}{j}\cdot \Prs{i}{\ell-j} \ .
\end{align}
Here $\Ppr{i}{j}$ denote the projection into the subspace in which
all the constraints in the $i$'th layer \emph{inside} the pyramid
have \emph{exactly} $j$ violations, and $\Prs{i}{j}$ denotes the
projection into the subspace in which the constraints of the $i$'th
layer \emph{outside} the pyramids have $j$ or less violations.

The core idea is that due to the pyramid structure, the support of
the constraints of pyramids at layer $i$ is included in the support
of the constraints of the pyramids at layer $i+1$, the layer beneath
it. Therefore we can ``pull back'' the $\Ppr{i'}{j'}$ operators
across the $\Prs{i}{j}$ operators as long as $i'<i$, and so
\begin{align}
\label{eq:pyr-decomp}
  \Pl{g}\cdots\Pl{1} = \sum_{j_1, \ldots, j_g} 
    (\Ppr{g}{j_g}\cdots\Ppr{1}{j_1}) \cdot 
    (\Prs{g}{\ell-j_g}\!\!\!\cdots\Prs{1}{\ell-j_1}) \ .
\end{align}
This is a central equation that will be useful for us in what to
follow.

\subsection{The exponential decay for the $2$ layers $\ell=0$ case}
\label{sec:2expo}

\begin{figure}
  \center \includegraphics[scale=0.7]{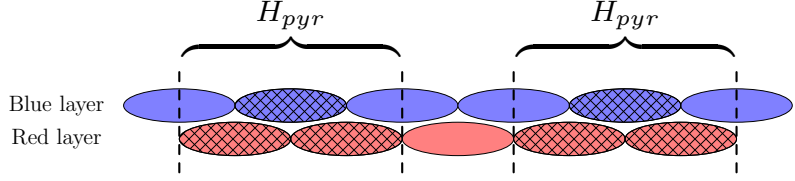} 
  \caption{The pyramids and associated $H_{pyr}$ spaces in a two-layers
  system \label{fig:pyramids-2layers}}
\end{figure}

To introduce the behavior of the exponential decay, we start with
the simpler case of two layers, which we call ``blue'' and ``red''.
The structure in this case is illustrated in
\Fig{fig:pyramids-2layers}.

We define $\Pi_{red}$ ($\Pi_{blue}$) to be the $\ell=0$ projections
from \Sec{sec:Pi}. This means that $\Pi_{red}$ projects into the 
tensor product of the zero (accepting) subspaces of all the terms in
the red  layer, and similarly the $P_{blue}$ for the blue layer.  

We may now write $\Pi_{red}$ and $\Pi_{blue}$ in terms of violations
inside and outside the pyramids as defined in \Eq{eq:inside-out}.
Notice that the $\ell=0$ case is particularly simple because
$\Pi_{red}$, $\Pi_{blue}$ can be written as products.  Take for
example $\Pi_{blue}$. It can the be written as the product
$\Pi_{blue} = (\Id-Q_1)\cdot(\Id-Q_2)\cdots$ with $Q_i$ being the
blue constraints. It is therefore clear that we can write
\begin{align}
  \Pi_{blue} &= \Delta_{blue} R_{blue} \ , 
\end{align}
where
\begin{align}
  \Delta_{blue} &= \text{terms inside the pyramids} \ , \\
  R_{blue} &= \text{terms outside the pyramids} \ .
\end{align}
Similarly, we define $\Pi_{red} = \Delta_{red}R_{red}$. 

As discussed in the previous section, because of the pyramids
structure, the support of $R_{red}$ and $\Delta_{blue}$ are
non-intersecting (See \Fig{fig:pyramids-2layers}) and therefore
\begin{align}
\label{eq:pullback}
  \Pi_{red}\Pi_{blue} = \Delta_{red}\Delta_{blue} R_{red} R_{blue} \ .
\end{align}

We now prove the exponential decay behavior.  Let us first
coarse-grain the $XY$ decomposition by gathering together all
sectors with the same number of $Y$ spaces. In other words, for
every integer $0\le s\le M$, define a projection
\begin{align}
  P_s \EqDef \sum_{|\nu|=s} P_\nu \ .
\end{align}
Then this is still a valid decomposition as the $P_s$ are orthogonal
to each other and $\sum_{s=0}^{m}P_s = \Id$.  The exponential decay
lemma states that if we apply
this decomposition to some state \emph{after} applying the
$\Pi_{blue}$ and $\Pi_{red}$ projections, then we can upper bound
the weight of the $s$ sector in terms of $\theta^s$.
\begin{lem}[Exponential-decay lemma for $\ell=0$]
\label{lem:expdecay2layers}
 
  Let $\ket{\psi}$ be an arbitrary (normalized) state, and consider
  the following normalized state  
  \begin{align}
    \ket{\Omega} \EqDef \frac{1}{x} \Pi_{red}\Pi_{blue}\ket{\psi} \ ,
  \end{align}
  and its coarse grained $XY$ decomposition
  \begin{align}
    \ket{\Omega} = \sum_s P_s \ket{\Omega} \EqDef \sum_s \lambda_s
      \ket{\Omega_s} \ .
  \end{align}
  Then there exist weights $\{\eta_s\}$ such that $\sum_s \eta^2_s
  \le 1$, and 
  \begin{align}
    \lambda_s \le
      \frac{1}{x}\theta^s\eta_s \ .
  \end{align}
\end{lem}
  
\begin{proof}
  To prove this claim, we take one step backwards,
  and write $\ket{\Omega}$ in terms of the fine-grained $XY$
  decomposition: $\ket{\Omega} = \sum_\nu \lambda_\nu
  \ket{\Omega_\nu}$. Then 
  \begin{align}
    \lambda_\nu^2 &= \bra{\Omega} P_\nu \ket{\Omega} \\
      &=\frac{1}{x^2} \bra{\psi}\Pi_{blue}\Pi_{red}\ 
        P_\nu \ \Pi_{red}\Pi_{blue}\ket{\psi} \ .
  \end{align}
  We now use \Eq{eq:pullback} and write
  \begin{align}
    \Pi_{blue}\Pi_{red}\ P_\nu \ \Pi_{red}\Pi_{blue} =
      R_{blue}R_{red}\Delta_{blue}\Delta_{red}\ P_\nu \
      \Delta_{red}\Delta_{blue}R_{red}R_{blue}  \ ,
  \end{align}
  and as $P_\nu$ commutes with $\Delta_{red}, \Delta_{blue}$,
  this is equal to
  \begin{align}
    R_{blue}R_{red}P_\nu \ \Delta_{blue}
    \Delta_{red}\Delta_{blue}\ P_\nu R_{red}R_{blue}  \ ,       
  \end{align}
  It follows that
  \begin{align}
    \lambda_\nu^2 \le \frac{1}{x^2} 
      \norm{P_\nu \Delta_{blue} \Delta_{red}\Delta_{blue}P_\nu}
      \cdot\norm{P_\nu \ket{\Phi}}^2 \ ,
  \end{align}
  with
  \begin{align}
    \ket{\Phi} \EqDef R_{red}R_{blue}\ket{\psi} \ .
  \end{align}
  
  Let us estimate $\norm{P_\nu \Delta_{blue}
  \Delta_{red}\Delta_{blue}P_\nu}$. Every operator in the product
  factors into a product of operators over every pyramid. Consider a
  pyramid site which is projected (by $P_{\nu}$) into a $Y$
  subspace. For brevity, call it a $Y$ pyramid.  Let $Q_{blue}$ be
  the blue constraint (the pyramid's top) and $Q_1, \ldots, Q_N$ the
  red constraints. We have 
  \begin{align}
    &P_Y \cdot Q_{blue}\cdot Q_1\cdots Q_N \cdot Q_{blue} \cdot P_Y \\
     &= (P_Y \cdot Q_{blue} \cdot Q_1\cdots Q_N \cdot P_Y) \cdot 
       (P_Y\cdot Q_N \cdots Q_1 \cdot Q_{blue} \cdot P_Y)
  \end{align}
  where we have used \Eq{eq:commutation1}. From
  \Eq{eq:theta}, its norm is smaller or equal to $\theta^2$, and
  since there are $|\nu|$ such $Y$ sites, we deduce that
  \begin{align}
    \norm{P_\nu \Delta_{blue}
      \Delta_{red}\Delta_{blue}P_\nu} \le \theta^{|2\nu|} \ .
  \end{align}
  All together, this leads to     
  \begin{align}
     \lambda_\nu^2 &\le \frac{1}{x^2} \theta^{2|\nu|} 
       \bra{\Phi} P_\nu\ket{\Phi} \ ,
  \end{align}
  and summing over all $\nu$ with $|\nu|=s$, we obtain
  \begin{align}
     \lambda_s^2 =\sum_{|\nu|=s}\lambda_\nu^2 
       \le \frac{1}{x^2} \theta^{2s} \eta_s^2
        \ ,
  \end{align}
  where we have defined
  \begin{align}
    \eta_s^2 \EqDef \bra{\Phi} P_s \ket{\Phi} 
      = \sum_{|\nu|=s} \bra{\Phi} P_\nu \ket{\Phi}\ .
  \end{align}
      \end{proof}

One can also make a similar statement for more than two layers; 
the proof follows very similar lines. Here we do not state 
this lemma since the following result implies it as a special case.

\subsection{Generalizing to many layers and $\ell>0$}
\label{sec:gen-decay}

In the general case we consider $g$ layers of constraints and $\ell$
might be larger than $0$. The projections $\Pi_{red}, \Pi_{blue}$
are replaced by the $\Pl{i}$ which project into the subspace of
$\ell$ or less violations in the $i$'th layer (see \Sec{sec:Pi}).
This allows us to derive an exponentially decaying bound on the
similarly defined coefficients $\lambda_s$, except now the bound will
contain some combinatorial factors depending on $\ell$. 
\begin{lem}[Exponential decay lemma for general $\ell$]
  \label{lem:gen-decay} Consider a $\kProj$ system with $g$ layers
  and $M$ projections, drawn from a finite set that is characterized
  by a parameter $0<\theta<1$. Let $0\le \ell\le M$ be an integer
  and let $\ket{\psi}$ be an arbitrary (normalized) state. Consider
  the following normalized state 
  \begin{align}
    \ket{\Omega} \EqDef \frac{1}{x} \Pl{g}\cdots\Pl{1}\ket{\psi} \ ,
  \end{align}
  and its coarse grained $XY$ decomposition
  \begin{align}
    \ket{\Omega} = \sum_s P_s \ket{\Omega} \EqDef \sum_s \lambda_s
      \ket{\Omega_s} \ .
  \end{align}
  Then there exist weights $\{\eta_s\}$ such that
  $\sum_s \eta^2_s \le 1$, and for every $s\ge \ell$, 
  \begin{align}
  \label{eq:genexpo}
    \lambda_s \le \frac{1}{x} k^{g^2\ell}
      \left(\frac{\ell+1}{\ell!}\right)^g
       s^{g\ell} \theta^s \eta_s \ .
  \end{align}
\end{lem}

The proof of the above claim is more involved than the 2-layers,
$\ell=0$ case, and will therefore be given in \App{sec:Pgen-decay}.
The main difficulty here is the fact that when $\ell>0$, the
projections $\Pl{i}$ are no longer a simple product of projections
as in the $\ell=0$ case. This can already be seen in
\Eq{eq:inside-out} that contains $\ell+1$ instead of the one term
that we find when we represent $\Pi_{red}$ or $\Pi_{blue}$ (the
$\ell=0$ case). Instead, they can be thought of as huge sum over
similar products of projections, where each such product projects
into a certain possible configuration of $\ell$ or less violations.
This complicates the analysis, but other than that, the proof
follows the same outline of the 2-layers, $\ell=0$ case.

\section{The detectability lemma for two layers and $\ell=0$}
\label{sec:2detect}

In this section we prove the detectability lemma for the special
case of two layers and $\ell=0$. The proof is considerably
simpler than the general proof,  yet it demonstrates the ideas of
the general case.

The general setup is similar to the one in \Sec{sec:2expo}. We
consider a $\kProj$ system with $\epsilon_0>0$ that can be arranged
in two layers. The first layer is called the ``blue layer''
and the second layer is the ``red layer'' (see
\Fig{fig:pyramids-2layers}). We define $\Pi_{red}$ as the projection
that projects into the accepting (zero) space of the red, and
similarly $\Pi_{blue}$. In addition, we assume that the constraints
are drawn from a finite family of constraints with a parameter
$0<\theta<1$ (see \Sec{sec:XY}). Then the 2-layers, $\ell=0$
detectability lemma is:
\begin{lem}[The detectability lemma for two layers and $\ell=0$]
\label{lem:detect-2layers}
  \ \\
  There exists a function $f(k)$ such that for every normalized
  state $\ket{\psi}$,
  \begin{align}
    \max\Big\{ \norm{ (\Id-\Pi_{red})\ket{\psi}}^2,
      \norm{ (\Id-\Pi_{blue})\ket{\psi}}^2\Big\} \ge
      \frac{1}{8}\Delta^2(0) \ ,
  \end{align}
  where 
  \begin{align}
  \label{eq:Delta0}
    \Delta^2(0) \EqDef 1 - \frac{1}{(\epsilon_0/f)
              \frac{(1-\theta^2)^3}{\theta^2} + 1} \ .
  \end{align}
\end{lem}

To prove this lemma, we will actually prove the following auxiliary lemma. 
\begin{lem}
\label{lem:aux2}
  For every normalized state $\ket{\psi}$, 
  \begin{align}
    \norm{ \Pi_{red}\Pi_{blue}\ket{\psi}}^2 \le 1 - \Delta^2(0) \ .
  \end{align}
\end{lem}

We will first show that \Lem{lem:detect-2layers} follows from
\Lem{lem:aux2}.

\begin{proof}[\ of \Lem{lem:detect-2layers} base on \Lem{lem:aux2}]

Assume \Lem{lem:aux2}, and assume by contradiction that both
$\norm{(1-\Pi_{red})\ket{\psi}}^2 < \frac{1}{8}\Delta^2(0)$ and
$\norm{(1-\Pi_{red})\ket{\psi}}^2 < \frac{1}{8}\Delta^2(0)$. Then we
write
\begin{align}
  \Pi_{blue}\Pi_{red}\Pi_{blue} = \Id + (\Pi_{blue}-\Id) + 
    \Pi_{blue}(\Pi_{red}-\Id) + \Pi_{blue}\Pi_{red}(\Pi_{blue}-\Id)
    \ ,
\end{align}
This way, except for the identity, every term on the RHS has a
$\Pi_{blue}-\Id$ or $\Pi_{red}-\Id$ on its right side. We want also
the left side to have such a term, so we continue in the same
fashion, and write:
\begin{align}
    \Pi_{blue}(\Pi_{red}-\Id) &= (\Pi_{blue}-\Id)(\Pi_{red}-\Id) +
      \Pi_{red}-\Id \ , \\
    \Pi_{blue}\Pi_{red}(\Pi_{blue}-\Id) &= (\Pi_{blue}-\Id)\Pi_{red}(\Pi_{blue}-\Id) +
      (\Pi_{red}-\Id)(\Pi_{blue}-\Id) + (\Pi_{blue}-\Id) \ .
\end{align}
All together we have
\begin{align}
  \Pi_{blue}\Pi_{red}\Pi_{blue} = \Id + \big[\text{$6$ terms with
  $(\Id-\Pi_{red})$ or $(\Id-\Pi_{red})$ on both sides.}\big]
\end{align}
When we ``sandwich'' the above equation with
$\bra{\psi}\cdot\ket{\psi}$, the absolute value of each of the $6$
terms will be smaller then $\Delta^2(0)/8$. This is due to the
Cauchy-Schwartz inequality and the assumption that
$\norm{(\Id-\Pi_{red})\ket{\psi}}$, $\norm{(\Id-\Pi_{blue})\ket{\psi}}
< \sqrt{\Delta^2(0)/8}$. Therefore,
\begin{align}
   \norm{ \Pi_{red}\Pi_{blue}\ket{\psi}}^2 &=
   \bra{\psi}\Pi_{blue}\Pi_{red}\Pi_{blue}\ket{\psi}
     > 1- \frac{6}{8}\Delta^2(0)
    > 1-\Delta^2(0) \ ,
\end{align}
which is a contradiction.

\end{proof}


We now proceed to prove \Lem{lem:aux2}.
\begin{proof}[\ of \Lem{lem:aux2}]
  Using the notation of \Sec{sec:2expo}, we define 
  \begin{align}
    \ket{\Omega} &\EqDef \frac{1}{x} \Pi_{red}\Pi_{blue}\ket{\psi} \ , \\
    x &\EqDef \norm{\Pi_{red}\Pi_{blue}\ket{\psi}} \ .
  \end{align}
  We wish to prove an upper bound for $x$.
   The idea of the proof is to estimate the total energy of
  $\bra{\Omega}E\ket{\Omega}$. This energy has no
  contributions from the red layer since $\ket{\Omega}$ has been 
  projected by $\Pi_{red}$, and so we may write:
  \begin{align}
    \epsilon_0 \le \bra{\Omega}
      \, E_{blue}\, \ket{\Omega} \ .
  \end{align}
  We will find an upper-bound for $\bra{\Omega} 
  \, E_{blue}\,\ket{\Omega}$ in terms of
  $\theta$, and this would give us an inequality for $x, \theta,
  \epsilon_0$. Inverting that inequality will give us the desired
  result.
  
  To estimate $E_{blue}$ we consider first one possible $XY$
  decomposition.  Let $E^{top}$ be the energy of all the blue
    constraints from the pyramids in this decomposition - the ``tops''
  of the pyramids. The main effort would be to find an upper-bound
  for $\bra{\Omega}E^{top}\ket{\Omega}$.  Once we do that, we can
  then repeat this process with other sets of pyramids (namely,
  other XY decompositions) until we cover all the blue constraints.
  All in all, there is a finite number $f(k)$ of $XY$ decompositions
  that are needed for that. Therefore,
  \begin{align}\label{eq:fksum}
    \epsilon_0 \le \bra{\Omega} E_{blue} \ket{\Omega} \le 
      f(k)\bra{\Omega} E^{top} \ket{\Omega}  \ . 
  \end{align}
  
  Hence, it remains to bound $\bra{\Omega} E^{top} \ket{\Omega}$. 
  We start by applying the fine- and coarse-grained $XY$
  decompositions  to $\ket{\Omega}$:
  \begin{align}
  \label{eq:2-Omega-XY}
    \ket{\Omega} = \sum_\nu \lambda_\nu \ket{\Omega_\nu} = \sum_s
    \lambda_s \ket{\Omega_s} \ .
  \end{align}
  Then as the $XY$ projections commute with the projections in
  $E^{top}$, we get
  \begin{align}\label{eq:2-Omega-cgXY}
    \bra{\Omega}E^{top}\ket{\Omega} = \sum_s \lambda_s^2
    \bra{\Omega_s} E^{top} \ket{\Omega_s}  \ .
  \end{align}
  
  \begin{claim}
  \label{cl:s-energy}
    \begin{align}
      \bra{\Omega_s} E^{top} \ket{\Omega_s}  \le s \ .
    \end{align}
  \end{claim}
    
  \begin{proof}
    We will prove this claim on the fine-grained $XY$ decomposition,
    by showing that $\bra{\Omega_\nu} E^{top} \ket{\Omega_\nu}
    \le |\nu|$. 
    
    Essentially, the claim follows from the fact that only the $Y$
    sites can contribute energy. Indeed, consider an $X$ pyramid,
    and let $Q$ be its blue constraint. Then by definition, either
    $Q\ket{\Omega_\nu} = 0$ or $Q\ket{\Omega_\nu} =
    \ket{\Omega_\nu}$. If the site contributes non-zero energy, the
    latter must hold. But $\ket{\Omega_{\nu}} \propto
    \Pi_\nu\ket{\Omega} \propto P_\nu
    \Pi_{red}\Pi_{blue}\ket{\psi}$, and so we get
    \begin{align}
    \label{eq:X-sector}
      QP_\nu \Pi_{red}\Pi_{blue}\ket{\psi} 
        = \Pi_\nu \Pi_{red}\Pi_{blue}\ket{\psi} \ .
    \end{align}
    We we show that the RHS of the above equation must vanish.
    Indeed, by \Eq{eq:pullback}, the LHS of the equation can be
    written as
    \begin{align}
      Q P_\nu \Delta_{red}\Delta_{blue} R_{red}R_{blue}\ket{\psi} \ .
    \end{align}
    But as the pyramids' projections commute with $P_\nu$, we get
    \begin{align}
      Q P_\nu \Delta_{red}\Delta_{blue} = 
      P_\nu Q \Delta_{red}\Delta_{blue} P_\nu \ ,
    \end{align}
    and because in the $X$ subspaces the blue and red constraints
    commute, this is equal to 
    \begin{align}
      P_\nu  \Delta_{red}Q\Delta_{blue} P_\nu \ .
    \end{align}
    This expression must vanish since $\Delta_{blue}$ contains a
    $\Id-Q$ term. It follows that the RHS of \Eq{eq:X-sector} must
    vanish and this proves the claim.
  \end{proof}

  We can now use the above bound inside \Eq{eq:2-Omega-cgXY},
  together with the bound 
  \begin{align}
  \label{eq:lambda-s}
    \lambda_s^2 \le \frac{1}{x^2} \theta^{2s}\eta_s^2 \
  \end{align} 
  which follows from 
  the exponential decay lemma~\ref{lem:expdecay2layers}. 
  We get: 
  \begin{align}\label{eq:bound} 
    \bra{\Omega}E^{top}\ket{\Omega} = \sum_s s\lambda_s^2 \le 
    \sum_s \frac{1}{x^2} s\theta^{2s}\eta_s^2 \ .
  \end{align} 

  In principle, inserting this into \Eq{eq:fksum} we could simply
  bound every $\eta_s^2$ by $1$, and, rearranging, get a bound on
  $x^2$. However, this bound would be bad for very small
  $\epsilon_0$.  Luckily, we can derive a stronger bound on
  $\eta^2_s$ for $s\ge 1$: 
  \begin{claim}\label{cl:etas} For every $s\ge 1$ 
    \begin{align}
      \eta^2_s \le \frac{1-x^2}{1-\theta^2} \ .
      \end{align}
  \end{claim} 
  \begin{proof}
  Summing over \Eq{eq:lambda-s}, we get
  \begin{align}
    1 \le \frac{1}{x^2} \sum_{s=0}^m \theta^{2s} \eta_s^2 \ ,
  \end{align}
  which is equivalent to
  \begin{align}
    x^2 \le \eta_0^2 + \sum_{s=1}^m \theta^{2s} \eta_s^2 \le
      \eta_0^2 + \theta^2\sum_{s=1}^m \eta_s^2 \ .
  \end{align}
  But $\sum_{s=0}^m \eta_s^2 \le 1$, so
  $\eta_0^2 \le 1-\sum_{s=1}^m \eta_s^2$, and
  \begin{align}
    x^2 \le 1-\sum_{s=1}^m \eta_s^2 + \theta^2\sum_{s=1}^m \eta_s^2 
     = 1 - (1-\theta^2)\sum_{s=1}^m \eta_s^2 \ ,
  \end{align}
  which leads to
  \begin{align}
    \sum_{s=1}^m \eta_s^2 \le \frac{1-x^2}{1-\theta^2} \ ,
  \end{align}
  implying the desired bound.
\end{proof}

We can now finish the proof of \Lem{lem:aux2}. Following \Eq{eq:bound} and
\Eq{eq:fksum} we have:

  \begin{align}
    \epsilon_0 \le \bra{\Omega} E_{blue} \ket{\Omega} \le 
      f(k)\cdot \bra{\Omega} E^{top} \ket{\Omega} 
      \le f(k)\cdot\frac{1-x^2}{x^2}\cdot\frac{\theta^2}{(1-\theta^2)^3}
      \ ,
  \end{align}
  which yields 
  \begin{align}
    x^2 \le \frac{1}{(\epsilon_0/f)\cdot\frac{(1-\theta)^3}{\theta^2} + 1}
     = 1 - \Delta^2(0) \ .
  \end{align}
\end{proof}


\section{The (general) detectability lemma}
\label{sec:gen-detect}

The detectability lemma can be generalized for more than 2 layers
and for $\ell>0$.  This generalization gives us a more detailed
picture of the energy distribution. This is important when
$\epsilon_0$ is much bigger than $1$ but is still smaller than its
maximal value $M$. In such a case, the detectability lemma asserts
that not only there exists a layer in which some violations are
detectable - but that there must be a layer in which $\ell$ or more
violations are detectable. In other words, it forbids a situation in
which in all the layers the violations are of only few constraints
and there is $1/poly$ weight on very high violations (so as to not
violate the minimal energy constraint). For the lemma to hold, we
need to require that $\ell$ -- the number of violations -- does not
exceed some normalized version of the minimal energy $\epsilon_0$.

\begin{lem}[The general detectability lemma]
\label{lem:detect} 

  Consider a $\kProj$ system with $g$ layers and a ground energy
  $\epsilon_0>0$. Let $\Pgt{i}$ denote a projection into the space
  of more than $\ell$ violations in the $i$'th layer. Then there
  exist integer functions $r(\theta, k, g), f(k,g)>1$ such that for
  every $0\le \ell <\frac{1}{r}
  \left(\frac{\epsilon_0}{f} - \frac{1}{1-\theta}\right)$ and 
  every normalized state $\ket{\psi}$ there is at least one layer
  $i$ in which:
  \begin{align}
    \norm{\Pgt{i} \ket{\psi}}^2 \ge \frac{1}{(2g)^2}\Delta^2(\ell) \ .
  \end{align}
  $\Delta(\ell)$ is a function of $\ell, \epsilon_0, \theta,
  k, g$, and is given by
  \begin{align}
  \label{def:Delta}
    \Delta^2(\ell) = \left\{
      \begin{array}{lcl}
        1 - \frac{1}{(\epsilon_0/f)
          \frac{(1-\theta^2)^3}{\theta^2} + 1} &,& \ell=0 \\
        1 - \frac{1}{1-\theta}\cdot
          \frac{1}{(\epsilon_0/f) - r\ell} &,& \ell>0
      \end{array}
    \right. \ .
  \end{align}
\end{lem}

The proof of the general detectability lemma is deduced from the 
exponential decay lemma for general $\ell$, using similar reasoning
to how the simpler detectability lemma is deduced from the $\ell=0$
exponential decay lemma. However, the technical details are much
more involved due to the same combinatorial factors that appear when
moving from $\ell=0$ to $\ell>0$ in the exponential decay lemmas.
The full proof in given in \App{sec:Pdetect}.

\ignore{
\begin{proof} {\bf Sketch:}
As in the 2-layers, $\ell=0$ case, this lemma is proved by proving
that for every state $\ket{\psi}$,
$\norm{\Pl{g}\cdots\Pl{1}\ket{\psi}}^2 \le 1-\Delta^2(\ell)$. The
proof of this fact again follows the same lines as the simple case.
We upper bound the energy of the normalized state $\ket{\Omega}
\EqDef \frac{1}{x}\Pl{g}\cdots\Pl{1}\ket{\psi}$ by a function of
$(x, \theta, k, g, \ell)$. The same energy is lower-bounded by
$\epsilon_0$, hence by inverting the inequality we obtain an upper
bound for $x^2$. To find the upper bound, the energy of each layer
is estimated by covering it with a finite number of $XY$
decomposition and estimating $E^{top}$ for each decomposition. This
is where we use the general exponential decay of \Sec{sec:gen-decay}.
\end{proof} }

\section{Relation of the simple detectability lemma and Kitaev's lemma}
\label{sec:Kitaev}

The $\ell=0$ detectability lemma for the two layers can be seen as
the converse of a special case of Kitaev's geometrical lemma, 
crucial in his proof of the quantum Cook-Levin theorem
\cite{ref:Kit02}.

\begin{lem}[Kitaev's lemma (see \Ref{ref:Kit02}]
  Given finite-dimensional operators $P\ge 0, Q\ge 0$ 
  with null eigenspaces, then
  \begin{align}
  \label{eq:kitaev}
    P + Q \ge \min\Big\{\Delta(P), \Delta(Q)\Big\}\cdot(1 - \cos\alpha) \ ,
  \end{align}
  where $\Delta(O)> 0$ is the smallest nonzero eigenvalue of $O$,
  and $\alpha$ the angle between the null spaces of $P$ and $Q$.
\end{lem}
Therefore if $\Delta(P), \Delta(Q)$ are fixed, by
lower-bounding $\alpha$, we can lower-bound the minimal energy of
$P+Q$. 

The 2-layers, $\ell=0$ detectability lemma can be seen as the
converse of this statement that holds in special case. In such case
let the $Q,P$ operators be the $\Pi_{red}$ and $\Pi_{blue}$
projections from \Sec{sec:2detect}. Then the detectability lemma can
be used to lower-bound $\alpha$. Indeed, for every state
$\ket{\psi}$, \Lem{lem:aux2} asserts that
\begin{align}
  \bra{\psi}\Pi_{blue} \Pi_{red} \Pi_{blue}\ket{\psi} 
    \le 1-\Delta^2(0) \ .
\end{align}
Then the angle $\alpha$ is given by
\begin{align}
  \cos\alpha = \min_{\substack{\ket{\psi}\in H_P\\ \norm{\psi}=1}}
    \bra{\psi}Q\ket{\psi}
    = \min_{\substack{\norm{\psi}=1}}
    \bra{\psi}\Pi_{blue} \Pi_{red}\Pi_{blue} \ket{\psi} 
      \le 1-\Delta^2(0) \ ,
\end{align}
where $H_P$ is the null space of $P$. 
Therefore, $1-\cos\alpha \ge \Delta^2(\ell=0,\epsilon_0)$, and combining
it with \Eq{eq:kitaev}, we get
\begin{align}
  \Delta^2(\ell=0,\epsilon_0) \le 1-\cos\alpha \le \epsilon_0 \ .
\end{align}
Moreover, looking at \Eq{def:Delta}, we see that in the limit
$\epsilon_0\to 0$, 
\begin{align}
  \left(\frac{1-\theta^2}{\theta}\right)^2
    (1-\theta)\frac{\epsilon_0}{f} \le 1-\cos\alpha \le \epsilon_0 \ .
\end{align}

\section{The quantum gap amplification lemma}
\label{sec:gap}
Below we describe first the well known classical setting 
of gap amplification using walks on expander graphs
(for completeness 
we also provide a proof in the appendix). We then 
define and prove a quantum analogue of this lemma, using 
the machinery we have developed so far. 

\subsection{The classical amplification lemma on Expanders}
\label{sec:classic-gap}

We consider a $d$-regular expander graph $G=(V,E)$ with $n=|V|$
vertices and second largest eigenvalue $0<\lambda(G)<1$. With every
node of $G$ we associate a variable that takes values in a finite
alphabet $\Sigma$.  Every edge is associated with a local constraint
on the two values of the nodes in the edge. We refer to the set of
constraints as a constraint system $\mathcal{C}$. 

Let $\sigma$ denote an assignment of the variables. We define
$\UNSAT_\sigma(G)$ to be the fraction of unsatisfied edges for that
under that assignment: 
\begin{align}
  \UNSAT_\sigma(G) = \frac{\text{\# of unsatisfied edges}}{|E|} \ .
\end{align}

In the amplification lemma, we define a new constraint system on $G$
using the notion of a $t$-walk. A $t$-walk on a graph $G$ is a
sequence of $t+1$ adjacent vertices, corresponding to a path of $t$
steps on $G$, starting at the vertex $v_0$ and ending at $v_t$. We
denote the edges along the path by $\Be = (e_1, e_2, \ldots, e_t)$.
The new constraint system is defined as follows. Consider all
possible $t$-walks on $G$, and for each $t$-walk $\Be=(e_1, \ldots,
e_t)$ we define a constraint that is satisfied if and only if all
the constraints along the path are satisfied. Notice that the new
constraints are less local than the original constraints, since they
are defined on up to $t+1$ vertices. Moreover, the new constraint
system can no longer be thought of as a ``constraint-graph'' since
its constraints are no longer defined on edges but on sets of $t+1$
nodes. Rather, it is a constraint ``hyper-graph''. With some abuse
of notation, we will call the new constraint system $G^t$.

The $\UNSAT$ of $G^t$ is defined by
\begin{align}
  \UNSAT_\sigma(G^t) = \frac{\text{\# of unsatisfied
  $t$-walks}}{\text{total \# of $t$-walks}} \ .
\end{align}
It seems plausible that $\UNSAT_\sigma(G^t)$ would be significantly 
larger than $\ge \UNSAT_\sigma(G)$, since is one edge is unsatisfied in 
$G$, it would appear in many $t$-walks in $G^t$. If the constraints
in $G^t$  
were chosen by choosing $t$ edges in $G$ independently,
then we would have expected an amplification factor of $t$. 
The fact that we consider walks on an expander means that 
the behavior is very similar to the completely random case. 
The amplification lemma thus shows that that by moving from $G$ to $G^t$, the
$\UNSAT$ is ``amplified'' by a factor of $t$, provided that
$\UNSAT(G)$ is not too close to $1$.

\begin{lem}[The classical amplification lemma]\label{lem:classicalGA} 
  Let $G=(V,E)$ be an expander graph with second largest eigenvalue
  $0<\lambda<1$, and let $\mathcal{C}$ be a constraint system on it
  using an alphabet $\Sigma$. Let $G^t$ denote the $t$-walk
  constraint system that was defined above. 
  Define 
  \begin{align}
  \label{def:c}
    c(\lambda) \EqDef \frac{1}{2+\frac{2}{1-\lambda}}.  
  \end{align}
Then for every
  assignment $\sigma$, 
  \begin{align}
    \UNSAT_\sigma(G^t) \ge \left\{
      \begin{array}{lcl}
        t\cdot c(\lambda)\cdot \UNSAT_\sigma(G) &,& 
          \UNSAT_\sigma(G) \le \frac{1}{t} \\
        c(\lambda)  &,& 
          \UNSAT_\sigma(G) \ge \frac{1}{t} \ .
      \end{array}  \right.
  \end{align}
\end{lem}

The proof is provided in Appendix \ref{app:classicalGA}.  

\subsection{The quantum amplification lemma}
\label{sec:Pgap}

The setting of the quantum amplification lemma is a natural
generalization of the classical setting. We consider a $d$-regular
expander graph $G=(V,E)$ with a second-largest eigenvalue
$0<\lambda(G)<1$. On top of $G$ we define a $\kProj$ system as
follows. We identify every vertex with a qudit of dimension $q$.
Every edge $e\in E$ is identified with a projection $Q_e$ on the two
qudits that are associated with the vertices of the edge. This
defines $\kProj$ system with $k=2\log(q)$ and a Hamiltonian
\begin{align}
  H = \sum_{e\in E} Q_e \ .
\end{align}
For any state $\ket{\psi}$, we define the quantum $\UNSAT$ of the
system to be the average energy of the edges:
\begin{align}
  \QUNSAT_\psi(G) \EqDef \frac{1}{|E|}\bra{\psi} H\ket{\psi}
    = \frac{1}{|E|}\sum_{e\in E} \bra{\psi}Q_e\ket{\psi} \ .
\end{align}

To define a new -- ``amplified'' -- constraint system, we use a
construction similar to the classical case. We consider all possible
$t$-walks ($t$ is fixed) $\Be = (e_1, \ldots, e_t)$ and for each
such walk, we define a $t\log(q)$-local projection $Q_\Be$ as
follows. We take the intersection of all the accepting spaces along
the path and define it to be the accepting space of $Q_\Be$. In
other words, $Q_\Be$ projects into the orthogonal complement of that
space. We refer to the new system as $G^t$, and define
\begin{align}
  \QUNSAT_\psi(G^t) &\EqDef \frac{\sum_\Be \bra{\psi} Q_\Be
    \ket{\psi}} {\text{\# of $t$-walks}} \ , \\
  \QUNSAT(G^t) &\EqDef \min_{\psi} \QUNSAT_\psi(G^t) \ .  
\end{align}

As in the classical case, the quantum amplification lemma shows how
$\QUNSAT(G^t)$ is amplified with respect to $\QUNSAT(G)$. The
amplification is linear in $t$ when $\QUNSAT(G)$ is far enough from
$1$, and then becomes saturated, just like in the classical case. 

\begin{lem}[The quantum amplification lemma]
  Consider a $\kProj$ system on an expander graph $G=(V,E)$ with a
  second largest eigenvalue $0<\lambda<1$ as defined above. Then
  \begin{align}
    \QUNSAT(G^t) \ge c(\lambda)\cdot K(q,d, \theta) 
      \cdot\min\Big\{ t\cdot\QUNSAT(G), 1 \Big\} \ ,
  \end{align}
  Where $K(q, d, \theta)$ is independent of the
  graph size and $c(\lambda)$ is given by \Eq{def:c}.
\end{lem}

\begin{proof}

  By definition, $\QUNSAT(G)=\epsilon_0/|E|$ where $\epsilon_0$ is
  the ground energy of $G$. Let $\ket{\psi}$ be a state for which
  $\QUNSAT(G^t) = \QUNSAT_\psi(G^t)$. 
  
  We first notice that our $\kProj$ system can be written with at
  most $g = 2d$ layers. We choose a layer $i$ and expand
  $\ket{\psi}$ in terms of its violations in that layer:
  \begin{align}
    \ket{\psi} = \sum_{j=0}^{|E|} \alpha_j\ket{\psi_j} \ .
  \end{align}
  Here $\ket{\psi_j}$ is the projection of $\ket{\psi}$ to the space
  with $j$ violations in the $i$'th layer. Thus $\ket{\psi}$ is a
  superposition of states in which the number of violated
  constraints of the $i$'th level have a well-defined value. 
  \ignore{In each such state there are exactly $j$ places with
  violations and the rest are satisfied.} 
  
  We consider an auxiliary $\kProj$ system $G_i$ which has same
  underlying graph $G$ and the same constraints of the $i$'th layer
  - but the rest of the constraints are null - i.e. they are always
  satisfied. It is clear that for every state $\ket{\psi}$,
  $\QUNSAT_\psi(G^t) \ge \QUNSAT_\psi(G_i^t)$. Moreover, as all the
  projections in $G_i^t$ commute within themselves and with the
  original projections of the $i$'th layer, we have
  \begin{align}
  \label{eq:i-basis}
    \QUNSAT_\psi(G_i^t) = \sum_j \alpha_j^2 \cdot
      \QUNSAT_{\psi_j}(G_i^t) \ .
  \end{align}
  We will now show:
  \begin{claim}
    \begin{align}
    \label{eq:Qbound}
      \QUNSAT_{\psi_j}(G_i^t) \ge \left\{
        \begin{array}{lcl}
          t\cdot c(\lambda)\cdot \frac{j}{|E|} &,& 
            \text{for $j\le\frac{|E|}{t}$} \\
          c(\lambda) &,& 
            \text{for $j>\frac{|E|}{t}$} 
        \end{array}\right. \ .
    \end{align}
  \end{claim}

  \begin{proof}
    This follows from the classical amplification lemma.  We expand
    $\ket{\psi_j}$ as a superposition $\ket{\psi_j} = \sum_\nu
    \beta_\nu \ket{\psi_{\nu}}$, where $\ket{\psi_{\nu}}$ has a
    well-defined value ($1$ or $0$, namely violating or not) at each
    edge of $G_i$, with the total number of violations being exactly
    $j$. Moreover, it is easy to see that as the projection into the
    state $\ket{\psi_\nu}$ commutes with the projections of $G_i$,
    then
    \begin{align}
      \QUNSAT_{\psi_j}(G_i) &= \sum_\nu \beta_\nu^2\cdot
        \QUNSAT_{\psi_{\nu}}(G_i^t) \ , \\
      \QUNSAT_{\psi_j}(G_i^t) &= \sum_\nu \beta_\nu^2\cdot
        \QUNSAT_{\psi_{\nu}}(G_i^t) \ ,
    \end{align}
    hence it is sufficient to prove
    \Eq{eq:Qbound} for
    $\QUNSAT_{\psi_{\nu}}(G_i^t)$. This, however, follows directly
    from the classical amplification lemma since under the state
    $\ket{\psi_{\nu}}$ the constraints of $G_i$ have a
    well-defined, classical values. We can therefore treat the situation 
    as a
    classical system $G_c$ with some assignment $\sigma$ and
    $\UNSAT_\sigma(G_c) = j/|E|$. According to
    the classical amplification lemma, if $j/|E| \le 1/t
    \Leftrightarrow j \le |E|/t$ then $\UNSAT_\sigma(G^t_c) \ge
    t\cdot c(\lambda)\cdot \frac{j}{|E|}$, otherwise,
    $\UNSAT_\sigma(G^t_c) \ge c(\lambda)$. But as everything is
    classical for $G_i$ and $G_i^t$ in the $\nu$
    sector then,
    \begin{align}
      \UNSAT_\sigma(G^t_c) = \QUNSAT_{\psi_{\nu}}(G_i^t)
    \end{align}
    and this proves the claim.
  \end{proof}
    
  Let us now use this claim to estimate the amplification. Combining
  \Eq{eq:Qbound} with \Eq{eq:i-basis}, we find
  \begin{align}
  \label{eq:UNSAT-Gt}  
    \QUNSAT(G^t) &= \QUNSAT_\psi(G^t) \ge \QUNSAT_\psi(G_i^t) \\
      &\ge t\frac{c(\lambda)}{|E|}\left( \alpha_1^2 + 2\alpha_2^2 + 3\alpha_3^2
      + \ldots + \frac{|E|}{t}\alpha_{|E|/t}^2 \right) +
       c(\lambda)\left( \alpha^2_{|E|/t+1} + \ldots+
       \alpha^2_{|E|}\right) \ . 
  \end{align}
  Therefore, as $\QUNSAT(G)=\frac{\epsilon_0}{|E|}$, the
  amplification ratio we are looking for is
  \begin{align}
    \label{eq:main-amp}
    \frac{\QUNSAT(G^t)}{\QUNSAT(G)} \ge
      t\frac{c(\lambda)}{\epsilon_0}\left( \alpha_1^2 + 2\alpha_2^2 + 3\alpha_3^2
      + \ldots + \frac{|E|}{t}\alpha_{|E|/t}^2 \right) +
        c(\lambda)\cdot\frac{|E|}{\epsilon_0}
          \left( \alpha^2_{|E|/t+1} + \ldots+
       \alpha^2_{|E|}\right)
  \end{align}
  The above equation is central and can be derived for \emph{any}
  layer (namely, for any $i$).  However, without additional
  information, it cannot be used to show amplification of
  $\QUNSAT(G^t)$. The reason is that the weights $\alpha_j^2$ can
  theoretically conspire in such a way that no amplification would
  occur. For example, $1/poly(|E|)$ of the weight can be
  concentrated on $\alpha^2_{|E|}$ and the rest on $\alpha^2_0$, and
  then there is no amplification since in these two sectors there is
  no amplification (one is completely satisfied and the other is
  completely saturated). Fortunately, we can use the detectability
  lemma to rule out the possibility that this sort of non-amplifying
  distribution appears \emph{simultaneously} in all layers.

  The idea is to consider two possible cases: $\frac{\epsilon_0}{f}
  - \frac{4}{1-\theta} \le 2r$ (the low-energy case) and
  $\frac{\epsilon_0}{f} - \frac{4}{1-\theta} > 2r$ (the high-energy
  case). For the former we use the $\ell=0$. In the former, we use
  the $\ell>0$ detectability lemma.  Let us start with the low
  energy case.
  
  %
  \subsection{The low energy case: $\frac{\epsilon_0}{f} -
    \frac{4}{1-\theta} \le 2r$}
  
  Here we estimate the amplification using the $\ell=0$
  detectability. Specifically,  \Lem{lem:detect} ensures us that
  there a layer $i$ in which,
  \begin{align}
    \alpha_1^2 + \alpha_2^2 + \ldots + \alpha_{|E|}^2 \ge
    \frac{1}{(2g)^2}\Delta^2(0) \ .
  \end{align}
  On the other hand, it is easy to see that \Eq{eq:main-amp} implies
  \begin{align}
    \frac{\QUNSAT(G^t)}{\QUNSAT(G)} &\ge
      t\frac{c(\lambda)}{\epsilon_0}\left( \alpha_1^2 + \alpha_2^2 + \alpha_3^2
      + \ldots + \alpha_{|E|}^2 \right) \ .
  \end{align}
  Therefore, 
  \begin{align}
    \label{eq:ell0-amp}
    \frac{\QUNSAT(G^t)}{\QUNSAT(G)} 
      \ge  t\cdot c(\lambda)\cdot (2g)^{-2}\cdot
      \frac{\Delta^2(0)}{\epsilon_0} \ .
  \end{align}
  
  Let us now lower bound the expression
  $\frac{\Delta^2(0)}{\epsilon_0}$. $\Delta^2(0)$ is a continuous
  function of $\epsilon_0$ that is bounded between $0$ and $1$ for
  $\epsilon_0\ge 0$. We have to worry about to things: (i) if
  $\epsilon_0$ becomes too large, the ratio might become small, and
  (ii) as $\epsilon_0\to 0$, also $\Delta^2(0) \to 0$. The first
  worry is taken cared by fact that in the low-energy case
  $\epsilon_0$ is upper bounded by $\frac{\epsilon_0}{f} -
  \frac{4}{1-\theta} \le 2r$. The second one is taken cared by
  noticing the approach of $\Delta^2(0)$ to $0$ as $\epsilon_0$ is
  linear in $\epsilon_0$ (see \Eq{def:Delta}). Therefore as
  $\epsilon_0\to 0$, the ratio approaches some positive constant. 
  All in all, we conclude that in the low-energy case,
  \begin{align}
    \frac{\QUNSAT(G^t)}{\QUNSAT(G)} \ge t\cdot c(\lambda)\cdot
      K_1(q,d, \theta) \ .
  \end{align}

  %
  \subsection{The high energy case: $\frac{\epsilon_0}{f} -
    \frac{4}{1-\theta} \ge 2r$}
    
  In the high-energy case, we use the detectability lemma with a
  particular $\ell$ to show the amplification. We choose $\ell$ as
  large as possible so that $(\ell+1)/\epsilon_0$ will be lower-bounded by
  a positive function of $q,d,\theta$. Specifically, 
  the high energy condition
  implies $\frac{\epsilon_0}{f} - \frac{2}{1-\theta} \ge 2r$, and
  so we choose\footnote{Note that by assumption
  $\epsilon_0$ is larger than $2r$, which can only happen when $|E|$
  -- the total number of constraints in the system -- satisfies
  $|E|>2r$, therefore the $\ell$ we choose makes sense.}
  \begin{align}
  \label{eq:choose-ell}
    \ell = \left\lfloor \frac{1}{r}\left(\frac{\epsilon_0}{f} -
    \frac{2}{1-\theta}\right)\right\rfloor \ge 2 \ .
  \end{align}
  Then on one hand, 
  \begin{align}
    \ell < \frac{1}{r}\left(\frac{\epsilon_0}{f} -
    \frac{2}{1-\theta}\right) \ ,
  \end{align}
  and so $(1-\theta)\left(\frac{\epsilon_0}{f} -r\ell\right) > 2$,
  yielding a finite detectability in \Lem{lem:detect}:
  \begin{align}
    \Delta^2(\ell) > 1-\frac{1}{2} &= \frac{1}{2} \\
    &\Downarrow \\
    \alpha_{\ell+1}^2 + \alpha_{\ell+2}^2 + \ldots + \alpha_{|E|}^2
      &\ge \frac{1}{(2g)^2}\Delta^2(\ell) \ge \frac{1}{8g^2} \ .
  \end{align}
  On the other hand, \Eq{eq:choose-ell} also implies 
  \begin{align}
    \ell+1 \ge \frac{1}{r}\left(\frac{\epsilon_0}{f} -
    \frac{2}{1-\theta}\right) \ ,
  \end{align}
  and so
  \begin{align}
  \label{eq:mid}
    \frac{\ell+1}{\epsilon_0} \ge \frac{1}{r}\left(\frac{1}{f} -
    \frac{2}{\epsilon_0(1-\theta)}\right) \ .
  \end{align}
  But $\frac{\epsilon_0}{f} - \frac{4}{1-\theta} > 0$, therefore
  \begin{align}
  \label{eq:ell1}
    \frac{\ell+1}{\epsilon_0} \ge \frac{1}{2fr} \ .  
  \end{align}
  
  Let us now return to \Eq{eq:main-amp}. By omitting all the
  $\alpha_i^2$ terms with $i\le \ell$, we obtain
  \begin{align}
    \frac{\QUNSAT(G^t)}{\QUNSAT(G)} \ge
      t\frac{c(\lambda)}{\epsilon_0}(\ell+1)
        \left( \alpha^2_{\ell+1} + \ldots  \alpha_{|E|/t}^2 \right) +
        c(\lambda)\cdot\frac{|E|}{\epsilon_0}
          \left( \alpha^2_{|E|/t+1} + \ldots + \alpha^2_{|E|}\right)
  \end{align}
  Define
  \begin{align}
    A &= \alpha_{\ell+1}^2 + \ldots + \alpha_{|E|/t}^2 \ , \\
    B &= \alpha_{|E|/t+1}^2 + \ldots + \alpha_{|E|}^2 \ .
  \end{align}
  Then $A+B \ge \frac{1}{8g^2}$ and
  \begin{align}
      \frac{\QUNSAT(G^t)}{\QUNSAT(G)} \ge
      t\frac{c(\lambda)}{\epsilon_0}(\ell+1) A + 
        c(\lambda)\cdot\frac{|E|}{\epsilon_0} B \ .
  \end{align}
  If $A\ge \frac{1}{16g^2}$ then from the above equation and by
  \Eq{eq:ell1},
  \begin{align}
    \frac{\QUNSAT(G^t)}{\QUNSAT(G)} \ge
    t\cdot c(\lambda)\cdot\frac{\ell+1}{\epsilon_0} A 
    \ge t\cdot c(\lambda)\cdot\frac{1}{2fr}\cdot\frac{1}{16g^2}
    \EqDef  t\cdot c(\lambda)\cdot K_2(q,d,\theta) \ .
  \end{align}
  If, on the other hand, $B\ge \frac{1}{16g^2}$ then we can use
  \Eq{eq:UNSAT-Gt} to conclude that
  \begin{align}
    \QUNSAT(G^t) \ge c(\lambda) \cdot \frac{1}{16g^2} \EqDef
    c(\lambda)\cdot K_3(d) \ .
  \end{align}

  Combining all these 3 results, it is straightforward to define a
  function $K(q,d,\theta)$ such that
  \begin{align}
        \QUNSAT(G^t) \ge c(\lambda)\cdot K(q,d, \theta) 
      \cdot\min\Big\{ t\cdot\QUNSAT(G), 1 \Big\} \ .
  \end{align}
\end{proof}

\section{Discussion regarding quantum PCP}\label{sec:qpcp}
We discuss here quantum PCP in the context of gap 
amplification.  
To this end we define what we
mean by a quantum PCP theorem. 

The classical PCP theorem can be viewed as a strong characterization
of the $\NP$ class. One way to state it is by first defining the
class $\PCP[r,q]$. This is the class of all languages $L$ for which
there is a polynomial verifier that uses $\orderof{r}$ random bits
and has the following properties. It reads an instance $x$ and has
an oracle access to $\orderof{q}$ bits of some proof $\pi$. If $x\in
L$ there is a witness for which the verifier accepts with
probability 1. Otherwise, for every proof, the acceptance
probability is smaller than $1/2$. The PCP theorem then states that
$\NP = \PCP[\log(n), 1]$. 

To state the quantum PCP conjecture we first recall  the
quantum analogous of the $\NP$ class - the $\QMA$ class.
\begin{deff}[The class $\QMA$]
  A language $L$ is in $\QMA$ if there exists a quantum polynomial
  verifier $V$ and a polynomial $p(\cdot)$ such that
  \begin{itemize}
    \item If $x\in L$, there exists a witness $\ket{\xi}\in
      \Qbit^{\otimes p(|x|)}$ such that $\Pr[\text{$V(x,\ket{\xi})$
        accepts}] \ge 2/3$

    \item If $x\notin L$, then for every $\ket{\xi}\in
      \Qbit^{\otimes p(|x|)}$ we have $\Pr[\text{$V(x,\ket{\xi})$
        accepts}] \le 1/3$
  \end{itemize}
\end{deff}

Therefore, a natural definition for a  $\QPCP$ class is 
\begin{deff}[{The class $\QPCP[q]$}]
  A language $L$ is in $\QPCP[q]$ if there exists a quantum
  polynomial verifier $V$ and a polynomial $p(\cdot)$ with the
  following properties. $V$ receives as input a classical string $x$
  and a state $\ket{\xi}\in \Qbit^{\otimes p(|x|)}$. However, it has
  only access to $\orderof{q}$ random qubits from $\ket{\psi}$.
  In other words, it has access only to a a density matrix $\rho$
  which is the tracing out of all but the $\orderof{q}$ random qubits
  in $\ket{\xi}\bra{\xi}$. The random 
choice of the qubits is performed according 
to a probability distribution which is computed by the quantum verifier. 
We denote its action on $(x,\ket{\xi})$ by
  $V(x,\ket{\xi})$.

  Then the condition for $L$ to be in $\QPCP[q]$ is that 
  \begin{itemize}
    \item If $x\in L$, there exists a witness $\ket{\xi}\in
      \Qbit^{\otimes p(|x|)}$ such that $\Pr[\text{$V(x,\ket{\xi})$
        accepts}] \ge 2/3$.

    \item If $x\notin L$, then for every $\ket{\xi}\in
      \Qbit^{\otimes p(|x|)}$ we have $\Pr[\text{$V(x,\ket{\xi})$
        accepts}] \le 1/3$.
  \end{itemize}
\end{deff}
Notice that we did not give the quantum verifier any random bits,
since it is quantum and can generate randomness by itself. The above
definition can have various variants; for example, we might require
that the probability distribution, which defines which qubits the
verifier sees, is uniform. We do not dwell on the differences
between these definitions; they are subtle, and at this stage the
subject is not understood well enough (to us) in order to determine
the best definition. 

A quantum PCP theorem would read: 
\begin{conj}[Quantum PCP]
\label{con:qpcp}
  \begin{align}
    \QPCP[1] = \QMA \ .
  \end{align}
\end{conj}

An essentially equivalent way of formulating the quantum PCP theorem
is in terms of local Hamiltonians: is it possible to efficiently
transform any $\kProj$ system with $1/\poly$ promise gap into a
$\kProj$ system with constant promise gap. In the classical world,
this corresponds to the inapproximability of max-3SAT.

Recently, Dinur gave a beautiful new proof of the classical PCP
theorem \cite{ref:Din07}, which works directly in this setting. She
starts with a classical $\SAT$ system with a $1/\poly$ promise gap
and successively amplifies the gap by repeated doubling. This
doubling is accomplished by gap amplification followed by alphabet
reduction and degree reduction to control the size and locality. 

It is tempting to try to apply Dinur's proof to the quantum case,
with $\kProj$ replacing the $k$-$\SAT$ problem.  As mentioned in
\Sec{sec:gap}, the quantum UNSAT is the ground energy of the system
divided by the number of constraints; it is $\QMA$
complete\footnote{in this discussion we omit the important subtle
distinction between the notions of $\QMA$ and that of $\QMA^1$,
namely, $\QMA$ with one sided errors. This will be explained in a
later version.} to decide between the cases when it is zero or
larger than some threshold (called the promise gap) which is inverse
polynomial.  A quantum version of Dinur's approach would state that
this hardness holds even when the promise gap is constant.
Formally, this is stated as
\begin{conj}[Quantum PCP by gap amplification]
\label{con:qgap}
  There exists an efficient classical transformation that takes a
  $\kProj$ system with a promise gap of $1/\poly$ and transforms it
  into a new $\kProj$ system with a constant promise gap such that
  the original system has a zero ground energy iff the new system
  has a zero ground energy.
\end{conj}

By the $\QMA$-completeness of the $k$-$\QSAT$ problem, it is easy to
deduce that Conjecture~\ref{con:qpcp} follows from
Conjecture~\ref{con:qgap}.  Our quantum gap amplification lemma can
be seen as a step towards emulating Dinur's approach in the quantum
setting. 

We mention that it has been speculated that a quantum version of the
PCP theorem is impossible to achieve, at least along the lines of
Dinur's proof: Dinur's proof relies heavily on copying the values of
the nodes in the graph, whereas in the quantum setting such a
copying is impossible due to the no-cloning theorem, which asserts
that there is no unitary transformation that copies an unknown
state.  This argument seems problematic to formalize.  One of the
reasons is that the argument assumes that the transformation on the
Hamiltonian which amplifies the gap must be unitary.  However, there
is no such requirement on the Hamiltonian map.  In fact, we were
able to use this observation, and derive a quantum PCP theorem,
albeit with a \emph{doubly} exponential long proof, by a
straightforward discretization of the problem.  The resultant map on
Hamiltonians, and consequently on the eigenstates, is non-unitary
(not even a unitary embedding). On the other hand, it is not even
clear that unitary PCP transformations are ruled out.

We pose as an open problem to reduce the doubly exponential proof to
a singly exponential long proof quantum PCP; such a result would be
the quantum analogue of the early classical PCP results, in which
the proofs were of exponential size \cite{ref:Aro08}.

\section{Acknowledgments}
\label{sec:Acknowledgements}

We are grateful to Matt Hastings and Tobias Osborne for exciting and
inspiring discussions about the quantum gap amplification lemma, the
possibility of a quantum PCP theorem, and possible avenues to prove
(or disprove) it.  We are also grateful to Michael Ben-Or, Avinatan
Hassidim, and Barbara Terhal for useful discussions and comments.
Finally we thank Elad Eban for useful \LaTeX tricks.

%
%

\appendix

\section{Proving the exponential decay in the general case}
\label{sec:Pgen-decay}

In this section we prove the exponential decay in the general case,
which is stated in \Lem{lem:gen-decay} in \Sec{sec:gen-decay}. The
proof follows essentially the path of the 2-layers, $\ell=0$ case.
We start by proving the decay in the fine-grained $XY$
decomposition. Consider then a given $XY$ decomposition and some
sector $\nu$ with $|\nu|\ge \ell$.
\begin{claim}
\label{cl:weight} There exist $(\ell+1)^g$ states $\ket{\Phi_j},
  j=1, \ldots, (\ell+1)^g$ with $\norm{\Phi_j}\le 1$, such that the
  weight of every $XY$ sector $\nu$ with $|\nu|\ge\ell$, is
  bounded by
  \begin{align}
  \label{eq:main-bound}    
    \lambda_\nu^2 
      \le \frac{1}{x^2}\frac{(\ell+1)^g}{(\ell!)^{2g}}
       \left(|\nu|^{g\ell} k^{g^2\ell}\theta^{|\nu|}\right)^2
       \sum_{j=1}^{(\ell+1)^g} \norm{P_\nu \ket{\Phi_j}}^2 \ .
  \end{align}
  
\end{claim}
  
\begin{proof}
  By definition,
  \begin{align}
    \lambda^2_\nu &= \bra{\Omega} P_\nu\ket{\Omega} \\
      &= \frac{1}{x^2} \bra{\psi}\Pl{1}\cdots\Pl{g} 
        \cdot P_\nu\cdot
        \Pl{g}\cdots\Pl{1}\ket{\psi} \\
      &=\frac{1}{x^2} \norm{P_\nu\cdot
        \Pl{g}\cdots\Pl{1}\ket{\psi}}^2 \ .
  \end{align}
  
  Let us estimate $\norm{P_\nu\cdot \Pl{g}\cdots\Pl{1}\ket{\psi}}$.
  Using \Eq{eq:theta}, we find
  \begin{align}
    \norm{P_\nu \Pl{g}\cdots\Pl{1}\ket{\psi}} \le
      \sum_{j_1, \ldots, j_g} 
      \norm{P_\nu (\Ppr{g}{j_g}\cdots \Ppr{1}{j_1}) P_\nu} \cdot
      \norm{P_\nu(\Prs{g}{\ell-j_g}\cdots \Prs{1}{\ell-j_1})
      \ket{\psi}} \ .
  \end{align}
  
  We will upper-bound $\norm{P_\nu (\Ppr{g}{j_g}\cdots
  \Ppr{1}{j_g}) P_\nu}$. Every projection $\Ppr{i}{j}$ can be
  written as a sum of products of the form $Q\cdot Q \cdot (\Id-Q)
  \cdot \ldots$ that work on the projections of the $i$'th layer
  that are inside the pyramid, such that there are exactly $j$
  projections of the form $Q$ and the rest is of the form $\Id-Q$ -
  corresponding to exactly $j$ violations. 
  
  The product $\Ppr{g}{j_g}\cdots \Ppr{1}{j_g}$, therefore, contains
  a huge number of such products. However, when we ``sandwich'' it
  between two $P_\nu$ projections, only few survive - those that are
  compatible with the $X$ portion of $P_\nu$. Let us estimate how
  many survive in a given layer. The $X$ part is completely fixed,
  and therefore we have to choose from all the $Y$ projections at
  most $\ell$ violations. There are $|\nu|$ $Y$ sites and at each
  site there are at most $k^g$ constraints, so overall, for $\ell
  \le |\nu|$, the number of surviving constraints in a single layer
  is bounded by
  \begin{align}
    \binom{|\nu|k^g}{\ell} 
      \le \frac{1}{\ell!} (|\nu|k^g)^\ell \ .
  \end{align}
  Considering all $g$ layers, the total number of surviving terms is
  therefore bounded by $\left(\frac{1}{\ell!} |\nu|^\ell
  k^{g\ell}\right)^g$. The norm of each term is bounded by
  $\theta^{|\nu|}$ as there are $|\nu|$ $Y$ sites. Therefore, the
  overall norm is bounded by
  \begin{align}
    \norm{P_\nu (\Ppr{g}{j_g}\cdots \Ppr{1}{j_1}) P_\nu} 
        \le \frac{1}{(\ell!)^g}|\nu|^{g\ell} k^{g^2\ell}\theta^{|\nu|} \ .
  \end{align}

  Thus far, we got
  \begin{align}
    \norm{P_\nu \Pl{g}\cdots\Pl{1}\ket{\psi}} \le
     \frac{1}{(\ell!)^g}
     |\nu|^{g\ell} k^{g^2\ell}\theta^{|\nu|} 
     \sum_{j_1, \ldots, j_g} 
       \norm{P_\nu(\Prs{g}{\ell-j_g}\cdots \Prs{1}{\ell-j_g})
         \ket{\psi}} \ .
  \end{align}
  There are $(\ell+1)^g$ terms in that sum, and so using standard
  Cauchy-Schwartz argument we get
  \begin{align}
    \norm{P_\nu \Pl{g}\cdots\Pl{1}\ket{\psi}}^2 \le
     \frac{(\ell+1)^g}{(\ell!)^{2g}}
     \left(|\nu|^{g\ell} k^{g^2\ell}\theta^{|\nu|}\right)^2
     \sum_{j_1, \ldots, j_g}
       \norm{P_\nu(\Prs{g}{\ell-j_g}\cdots \Prs{1}{\ell-j_g})
         \ket{\psi}}^2 \ .
  \end{align}
  Finally, grouping all the indices $(j_1, \ldots, j_g)$ into one
  big index $j$, and defining the un-normalized states
  \begin{align}
    \ket{\Phi_j} \EqDef \Prs{g}{\ell-j_g}\cdots \Prs{1}{\ell-j_1}
         \ket{\psi} \ ,
  \end{align}
  whose norm is smaller than or equal to $1$, we get that for
  $|\nu|\ge \ell$, 
  \begin{align}
    \lambda_\nu^2 = \bra{\Omega} P_\nu \ket{\Omega} 
    \le \frac{1}{x^2}\frac{(\ell+1)^g}{(\ell!)^{2g}}
     \left(|\nu|^{g\ell} k^{g^2\ell}\theta^{|\nu|}\right)^2
     \sum_{j=1}^{(\ell+1)^g} \norm{P_\nu \ket{\Phi_j}}^2 \ .
  \end{align}
  
\end{proof}

To prove \Lem{lem:gen-decay}, pass to the coarse grained $XY$
decomposition by grouping together all the $XY$ sectors with the
same number of $Y$'s. 
Then 
\begin{align}
\label{eq:pre-gen-decay}
  \lambda_s^2 = \sum_{|\nu|=s} \lambda^2_\nu \le 
    \frac{1}{x^2}\frac{(\ell+1)^g}{(\ell!)^{2g}}
     \left(s^{g\ell} k^{g^2\ell}\theta^{s}\right)^2
     \sum_{j=1}^{(\ell+1)^g} \sum_{|\nu|=s} 
     \norm{P_\nu \ket{\Phi_j}}^2 \ .
\end{align}
Defining
\begin{align}
  \eta_s^2 \EqDef \frac{1}{(\ell+1)^g} 
    \sum_{j=1}^{(\ell+1)^g}
    \sum_{|\nu|=s} \norm{P_\nu \ket{\Phi_j}}^2
    = \frac{1}{(\ell+1)^g} 
    \sum_{j=1}^{(\ell+1)^g}
    \norm{P_s \ket{\Phi_j}}^2 \ ,
\end{align}
we find that $\sum_s \eta^2_s \le 1$ (recall that $\norm{\Phi_j} \le
1$) and by \Eq{eq:pre-gen-decay}, for every $s\ge \ell$, 
\begin{align}
  \lambda_s \le \frac{1}{x}k^{g^2\ell}
    \left(\frac{\ell+1}{\ell!}\right)^g
    s^{g\ell} \theta^s \eta_s\ .
\end{align}

\section{Proving general detectability lemma, \Lem{lem:detect}}
\label{sec:Pdetect}

To prove this lemma, we will prove the following auxiliary lemma
\begin{lem}
\label{lem:aux} Let $\Pl{i} = \Id - \Pgt{i}$ denote the
  projection into $\ell$ or less violations in the $i$'th layer as
  in \Sec{sec:gen-decay}. Then 
  \begin{align}
  \label{eq:x2}
    \norm{\Pl{g}\cdots\Pl{1}\ket{\psi}}^2 \le 1-\Delta^2(\ell) \ ,
  \end{align}
  where $\Delta(\ell)$ is defined in \Lem{lem:detect}, and in the
  $\ell>0$ case we assume that $(\epsilon_0/f) - r\ell >
  \frac{1}{1-\theta}$.
\end{lem}

The proof of \Lem{lem:aux} would be given later in \Sec{sec:Paux}.
Based on it, we can prove \Lem{lem:detect} as follows
\begin{proof}[\ of \Lem{lem:detect}]

  Given the state $\ket{\psi}$ and an integer $\ell\ge 0$, 
  assume that \Eq{eq:x2} holds and yet for every layer, 
  \begin{align}
  \label{eq:assumption}
    \norm{\Pgt{i}\ket{\psi}}^2 < \frac{1}{(2g)^2}\Delta^2(\ell)  \ .
  \end{align}
  For brevity, we denote
  \begin{align}
  \label{def:x}
    x \EqDef \norm{\Pl{g}\cdots\Pl{1}\ket{\psi}} \ .
  \end{align}
  Then
  \begin{align}
  \label{eq:x2-product}
    x^2 &= \bra{\psi} \Pl{1}\cdots
      \Pl{g-1}\Pl{g}\Pl{g-1}\cdots\Pl{1}\ket{\psi}  \ .
  \end{align}
    Every product of $N$ operators can be written as:
  \begin{align}
  \label{eq:expand}
    O_1\cdots O_N &= \Id + (O_1-\Id) + O_1(O_2-\Id) + O_1O_2(O_3-\Id)
      \\
      &+ \ldots + (O_1\cdots O_{N-1})\cdot(O_N-\Id) \ .
  \end{align}
  Expanding \Eq{eq:x2-product} this way, we get
  \begin{align}
    x^2 &= 1 + \bra{\psi}(\Pl{1} - \Id)\ket{\psi} + 
      \bra{\psi}\Pl{1}\big(\Pl{2} - \Id\big)\ket{\psi} +
      \bra{\psi}\Pl{1}\Pl{2}\big(\Pl{3} - \Id\big)\ket{\psi} +
      \ldots
  \end{align}
  The RHS of the above equation contains $2g-1$ terms of the form
  $\bra{\psi} \Pl{1}\cdots\Pl{i}\big(\Pl{i+1}-\Id\big)\ket{\psi}$.
  Let us estimate their magnitude. By an expansion similar to
  \Eq{eq:expand}, we write
  \begin{align}
    \Pl{1}\cdots\Pl{i} = (\Id-\Pl{i}) + (\Id-\Pl{i-1})\Pl{i} +
      \ldots \ .
  \end{align}
  Therefore $\bra{\psi}
  \Pl{1}\cdots\Pl{i}\big(\Pl{i+1}-\Id\big)\ket{\psi}$ can be written
  as a sum of at most $2g$ terms, each of them is an inner product
  of $\bra{\psi}(\Id-\Pl{j})$ times some projections, times
  $(\Id-\Pl{i})\ket{\psi}$. By our assumption, the norm of the ket
  and bra is smaller than $\Delta(\ell)/(2g)$ and as the norms of
  the projections are smaller than or equal to unity we find
  \begin{align}
    |\bra{\psi}
        \Pl{1}\cdots\Pl{i}\big(\Pl{i+1}-\Id\big)\ket{\psi}| \le 2g
        \frac{\Delta^2(\ell)}{(2g)^2} = \frac{\Delta^2(\ell)}{2g} \ .
  \end{align}
  Therefore, overall, 
  \begin{align}
    x^2 \le 1 + (2g-1)\frac{\Delta^2(\ell)}{2g} < 1 + \Delta^2(\ell)
    \ ,
  \end{align}
  contradicting \Eq{eq:x2}.
\end{proof}

We now turn to the proof of \Lem{lem:aux}. The outline of the proof
is very similar to the simple case of 2-layers, $\ell=0$, and was
discussed in \Sec{sec:gen-detect}. The main goal of the proof is to
estimate the energy of the normalized state
$\frac{1}{x}\Pl{g}\cdots\Pl{1}\ket{\psi}$, which The has
contributions from all layers. For every layer we will find a crude
upper bound of its energy as a function of $x$ (as well as of $\ell,
k, g, \theta$). Summing all these bounds together, we will get an
upper bound to the total energy. This energy is lower bounded by
$\epsilon_0$, the ground energy of the system, and this gives us an
inequality. We then reverse it and extract an upper bound for $x$.

We start by using the $XY$ decomposition to upper bound the energy
of the first layer.

%
\subsection{Estimating the energy of the first layer}

Consider then an $XY$ decomposition, and let $E^{top}$ denote the
energy of all the constraints of the first layer (the top layer in
\Fig{fig:pyramids}) that belong to the pyramids of the
decomposition. We define $\ket{\Omega}$ to be the following
normalized state: 
\begin{align}
\label{def:Omega}
  \ket{\Omega} &\EqDef \frac{1}{x}\Pl{g}\cdots\Pl{1}\ket{\psi} \ .
\end{align}
The entire section will be dedicated to proving the following lemma:
\begin{lem}
\label{lem:layer}
  For $\ell=0$,
  \begin{align}
  \label{eq:E1-ell0}
    \bra{\Omega}E^{top}\ket{\Omega} 
      \le  \frac{1-x^2}{x^2}
          \frac{\theta^2}{(1-\theta^2)^3} \ ,
  \end{align}
  and for $\ell>0$ there is a positive function function $r(\theta,
  k, g)$ (independent of $\ket{\Omega}$) such that
  \begin{align}
  \label{eq:E1-general}
    \bra{\Omega}E^{top}\ket{\Omega}
      \le r\ell + \frac{1}{x^2} \frac{1}{1-\theta} \ .
  \end{align}
\end{lem}

\noindent\textbf{Proof:}\\ 

The $\ell=0$ case was essentially already proved in the 2-layers
case in \Sec{sec:2detect} (specifically, see \Eq{eq:Delta0}). The
difference between the 2-layers case and the $g$-layers case are
semantic and therefore we will only consider the $\ell>0$ case. 

Consider the coarse- and fine-grained $XY$ decomposition of
$\ket{\Omega}$,
\begin{align}
  \ket{\Omega} = \sum_\nu \lambda_\nu \ket{\Omega_\nu} 
   = \sum_s \lambda_s\ket{\Omega_\nu} \ .
\end{align}
Since $E^{top}$ a sum of the inverses of pyramid projections from the
first layer, it must commute with the $XY$ projections $P_\nu$.
Therefore,
\begin{align}
  \bra{\Omega} E^{top}\ket{\Omega} 
    = \sum_s \lambda_s^2
      \bra{\Omega_s}E^{top}\ket{\Omega_s}  \ .
\end{align}
Our first claim is
\begin{claim}
\label{cl:energy}
  For every $s$ with non-zero weight $\lambda_s$,
  \begin{align}
    \bra{\Omega_s} E^{top}\ket{\Omega_s} \le \ell + s \ .
  \end{align}
\end{claim}

\begin{proof}
  It is sufficient to prove that for every sector $\nu$ with
  non-zero weight, $\bra{\Omega_\nu} E^{top}\ket{\Omega_\nu} \le
  \ell + |\nu|$.
  
  $E^{top}$ has one contribution from every pyramid top. Consider a
  sector $\nu$ of the fine-grained $XY$ decomposition. It contains
  $|\nu|$ $Y$ spaces and the rest are $X$ spaces. The maximal energy
  contribution from the $Y$ pyramids is therefore $|\nu|$. We will
  now show that the contribution from the $X$ pyramids is at most
  $\ell$. Essentially, the proof boils down to the fact that the
  projections commute on the $X$ sectors and therefore if
  $\ket{\Omega}$ has more than $\ell$ violations on the $X$ sectors
  then also $\Pl{1}\ket{\psi}$ has - which is impossible. The
  following argument shows this more formally.
  
  Let $Q_i$ be the projection in the first layer in the $i$'th
  pyramid, where $\nu$ has an $X$ sector. Then either
  $Q_i\ket{\Omega_\nu}=0$ or $Q_i\ket{\Omega_\nu} =
  \ket{\Omega_\nu}$. If the total contribution from all $X$ sectors
  is larger than $\ell$, there are $\ell+1$ pyramids in which
  $Q_i\ket{\Omega_\nu}=\ket{\Omega_\nu}$. For brevity, assume that
  these appear in the first $\ell+1$ pyramids. Then
  \begin{align}
    \big(\prod_{i=1}^{\ell+1}Q_i\big)\ket{\Omega_\nu} =
      \ket{\Omega_\nu} \ .
  \end{align}
  Assuming that $\lambda_\nu\ne 0$, we get
  \begin{align}
  \label{eq:proj-nu}
    \big(\prod_{i=1}^{\ell+1}Q_i\big)P_\nu\ket{\Omega} =
      P_\nu\ket{\Omega} \ .
  \end{align}
  Using the definition of $\ket{\Omega}$ in \Eq{def:Omega} and
  \Eq{eq:pyr-decomp}, the LHS of the above equation is equal to
  \begin{align}
    \frac{1}{x} \sum_{j_1, \ldots, j_g} 
      \big(\prod_{i=1}^{\ell+1}Q_i\big)P_\nu \cdot
      (\Ppr{g}{j_g}\cdots\Ppr{1}{j_g}) \cdot 
      (\Prs{g}{\ell-j_g}\!\!\!\cdots\Prs{1}{\ell-j_g})\ket{\psi} \ .
  \end{align}
  The above expression vanishes. The reason is that it is a sum over
  terms which all contain
  \begin{align}
    \big(\prod_{i=1}^{\ell+1}Q_i\big)P_\nu \cdot
      (\Ppr{g}{j_g}\cdots\Ppr{1}{j_g}) \ .
  \end{align}
  Using the fact that $P_\nu$ commutes with the constraints inside
  the pyramids, and that within an $X$ sector the constraints of the
  pyramids commute within themselves, this is equal to
  \begin{align}
    P_\nu \big(\prod_{i=1}^{\ell+1}Q_i\big)\cdot\Ppr{1}{j_{g-1}} \cdot
      (\Ppr{g}{j_g}\cdots\Ppr{1}{j_{g-1}}) P_\nu
  \end{align}
  But $\big(\prod_{i=1}^{\ell+1}Q_i\big)\cdot\Ppr{1}{j_{g-1}}$ is
  identically zero as it must contain at least one term of the form
  $Q(\Id-Q)$. It follows that $P_\nu \ket{\Omega} = 0$ which can
  only happen when $\lambda_\nu=0$.
\end{proof}

Next, we bound the weights $\lambda_s$ using the exponential decay
of \Sec{sec:gen-decay}, which is proved in \Sec{sec:Pgen-decay}.
According to \Lem{lem:gen-decay}, there exists a set of weights
$\eta^2_s$ such that $\sum_s \eta_s^2\le 1$ and for every $s\ge
\ell$, 
\begin{align}
  \lambda_s \le \frac{1}{x} k^{g^2\ell}
    \left(\frac{\ell+1}{\ell!}\right)^g
     s^{g\ell} \theta^s \eta_s  \ . 
\end{align}
Bounding  $\eta_s$ by $1$, we get that 
\begin{align}
  \lambda_s \le \frac{1}{x} k^{g^2\ell}
    \left(\frac{\ell+1}{\ell!}\right)^g
       s^{g\ell} \theta^s  \ .
\end{align}

We are now in position to prove the main result of this section,
\Eq{eq:E1-general}. In \App{sec:r} we use the above equation to show that
it is possible to find a constant $r(\theta, k, g)$ such that for
every $s > r\ell$,
\begin{align}
  \lambda_s^2 s \le \frac{1}{x^2}\theta^s\ .
\end{align}
Consequently, from \Cl{cl:energy} it follows that
\begin{align}
  \bra{\Omega} E^{top} \ket{\Omega} 
    &\le \ell + \sum_{s=0}^{r\ell} \lambda_s^2 s +
    \frac{1}{x^2}\sum_{s=r\ell+1}^\infty \theta^s \\
    &\le (r + 1)\ell 
      + \frac{1}{x^2} \frac{\theta^{(r\ell+1)}}{1-\theta} \ .
\end{align}
By redefining $r(\theta, k, g) \to r(\theta, k, g) + 1$, and using
the fact that $\theta^{(r\ell+1)}<1$, we recover \Eq{eq:E1-general}.

Finally, we can now prove \Lem{lem:aux}.

\subsection{Proof of \Lem{lem:aux}}
\label{sec:Paux}

To prove \Lem{lem:aux}, we use \Lem{lem:layer} to estimate the
total energy of the system. To estimate the energy of the first
layer, we apply \Lem{lem:layer} several times using different $XY$
decomposition. The $XY$ decompositions are chosen such that every
constraint in the first layer appears in exactly one $XY$
decomposition. One can easily verify that the total number of such
decompositions that is needed for this task is upper bounded by some
constant $f_1(k,g)$. Therefore,
the total energy of the first layer is bounded by
\begin{align}
  \bra{\Omega} E_1 \ket{\Omega} \le \left\{
    \begin{array}{lcl}
      f_1 \frac{1-x^2}{x^2} 
        \frac{\theta^2}{(1-\theta^2)^3} &,& \ell=0 \\
      f_1 \left[r\ell 
      + \frac{1}{x^2} \frac{1}{1-\theta} \right] 
        &,& \ell>0
    \end{array}  
  \right.
\end{align}
To bound the energy of the second layer, we can apply the derivation
of the first layer, with some trivial modifications:
\begin{align}
  g &\to g-1 \ , \\
  \ket{\psi} &\to \Pl{1}\ket{\psi} \ .
\end{align}
In addition, we need to update the functions $r(\theta, k, g)$ and
$f_1(k,g)$. It is easy to see that both of them can be decreased.
Therefore, it is not surprising to see that the upper bound of the
$\bra{\Omega}E_2\ket{\Omega}$ is smaller than the upper bound of
$\bra{\Omega}E_1\ket{\Omega}$, and this true for all
the other layers. Consequently, by setting $f(k,g) \EqDef
gf_1(k,g)$, we get:
\begin{align}
  \epsilon_0 \le \bra{\Omega} E \ket{\Omega} \le \left\{
    \begin{array}{lcl}
      f \frac{1-x^2}{x^2} 
        \frac{\theta^2}{(1-\theta^2)^3} &,& \ell=0 \\
      f \left[r\ell 
      + \frac{1}{x^2} \frac{1}{1-\theta} \right] 
        &,& \ell>0
    \end{array}  
  \right. \ .
\end{align}
Here $\epsilon_0$ is the ground energy of the system.
\Lem{lem:aux} is now proved by inverting this inequality:
\begin{align}
  \text{for $\ell=0$:} \quad & x^2 \le 
        \frac{1}{(\epsilon_0/f)
          \frac{(1-\theta^2)^3}{\theta^2} + 1} 
           = 1 - \Delta^2(0)\ , \\
  \text{for $\ell>0$:} \quad & x^2 \le
        \frac{1}{1-\theta}\cdot
        \frac{1}{(\epsilon_0/f) - r\ell} = 1 - \Delta^2(\ell) \ .
\end{align}
Note, of course, that the $\ell>0$ inequality is only valid for 
$(\epsilon_0/f) - r\ell > 0$.

\section{Finding $r(\theta,k,g)$}
\label{sec:r}

In this section we prove that it is possible to find a constant
$r(\theta, k, g)$ such that for every $s > r\ell$, 
\begin{align}
\label{eq:cond-again}
  \frac{1}{x^2}\left(\frac{\ell+1}{\ell!}\right)^{2g}
      k^{2g^2\ell}s^{2g\ell+1} \theta^{2s} \le \frac{1}{x^2}\theta^s\ .
\end{align}
Eliminating a factor $\frac{\theta^s}{x^2}$ and taking a $\log$ of
the equation, we find the following sufficient condition
\begin{align}
  2g[\log(\ell+1)-\log(\ell!)] + 2g^2\ell\log(k) + (2g\ell+1)\log(s) +
  s\log(\theta) < 0
\end{align}
which is equivalent to 
\begin{align}
  \frac{2g}{s}[\log(\ell+1)-\log(\ell!)] + 2g^2\frac{\ell}{s}\log(k)
   + (2g\ell+1)\frac{\log(s)}{s} + \log(\theta) < 0
\end{align}
Re-arranging it gives
\begin{align}
  \frac{2g}{s}[\ell\log(s)-\log(\ell!)] + \frac{2g}{s}\log(\ell+1)
  + 2g^2\frac{\ell}{s}\log(k) + \frac{\log(s)}{s}  < \log(1/\theta)
  \ .
\end{align}
On the LHS we have the sum of 4 terms. For $s>3$, $\log(s)/s$ is
monotonically decreasing ($\log(\cdot)$ is the natural logarithm).
So if the above condition holds for $s=r\ell$ with $r>3$, it would
hold for any $s>r\ell$. Therefore, a sufficient condition is
\begin{align}
  \frac{2g}{s}[\ell\log(s)-\log(\ell!)] + \frac{2g}{r\ell}\log(\ell+1)
    + \frac{2g^2\log(k)}{r} + \frac{\log(r)}{r}  < \log(1/\theta)
  \ .
\end{align}
Next, for every $\ell\ge 1$, the term $\frac{\log(\ell+1)}{\ell}$,
which appears in the second element is smaller than $1$, therefore a
sufficient condition is
\begin{align}
  \frac{2g}{s}[\ell\log(s)-\log(\ell!)] + \frac{2g}{r}
    + \frac{2g^2\log(k)}{r} + \frac{\log(r)}{r}  < \log(1/\theta)
  \ .
\end{align}
Let us now analyze the first term. Using Sterling's approximation, we
get
\begin{align}
  \frac{2g}{s}[\ell\log(s)-\log(\ell!)] &\le
  \frac{2g}{s}[\ell\log(s)-\ell\log(\ell) +\ell ] \\
    &=\frac{2g}{s/\ell}\log(s/\ell) + \frac{2g}{s/\ell} \ .
\end{align}
Again, using the assumption that $s/\ell > r > 3$, then
$\frac{\log(s/\ell)}{s/\ell} < \log(r)/r$, and $\frac{2g}{s/\ell} <
2g/r$. So overall, we find that as long as $r>3$, a sufficient
condition for \Eq{eq:cond-again} is
\begin{align}
  (2g+1)\frac{\log(r)}{r} 
    + \frac{4g+2g^2\log(k)}{r} < \log(1/\theta) \ .
\end{align}
The LHS of the above inequality approaches zero as $r\to+\infty$,
hence we can find an $r(\theta, k, g)>3$ that satisfies it.

\section{Proof of Classical Amplification lemma}\label{app:classicalGA} 

\begin{proof}(of Lemma \ref{lem:classicalGA}) 
  Given the assignment $\sigma$, we let $F\subseteq E$ denote the
  set of unsatisfied edges in $G$. Obviously, $\UNSAT_\sigma(G) =
  \frac{|F|}{|E|}$. Consider the homogeneous probability
  distribution over all $t$-walks. We define a random variable
  $Z(\Be)$ that counts the number of unsatisfied edges in the
  $t$-walk $e=(e_1, \ldots, e_t)$. Then 
  \begin{align}
    \UNSAT_\sigma(G^t) = \Pr[Z(\Be) > 0] \ .
  \end{align}
  Moreover, since $Z(\Be)$ is a non-negative random variable that is
  not identically $0$, 
  \begin{align}
  \label{eq:positive}
     \Pr[Z(\Be) > 0] \ge \frac{\Av^2[Z(\Be)]}{\Av[Z^2(\Be)]}  \ .
  \end{align}
  In what follows, we will lower-bound $\Av[Z(\Be)]$ and upper-bound
  $\Av[Z^2(\Be)]$. To do that, we write $Z(\Be) = \sum_{i=1}^t
  Z_i(\Be)$, where $Z_i(\Be)$ is the random variable that is equal
  to $1$ if the $i$'th edge of $\Be$ is unsatisfied and to $0$
  otherwise. It is easy to see that for every $i$, $\Av[Z_i(\Be)] =
  |F|/|E|$, and therefore 
  \begin{align}
    \Av[Z(\Be)] = t\frac{|E|}{|F|} \ .
  \end{align}
  
  To bound $\Av[Z^2(\Be)]$, we write
  \begin{align}
    \Av[Z^2(\Be)] = \sum_{i,j} \Av[Z_i(\Be)Z_j(\Be)] = 
    \sum_{i=1}^t \Av[Z^2_i(\Be)] + 2\sum_{i<j} \Av[Z_i(\Be)Z_j(\Be)]
    \ .
  \end{align}
  Note that $Z_i^2(\Be)=Z_i(\Be)$ and so
  \begin{align}
  \label{eq:Z2}
    \Av[Z^2(\Be)] = t\frac{|F|}{|E|} + 2\sum_{i<j}
      \Av[Z_i(\Be)Z_j(\Be)]   \ .
  \end{align}
  To estimate $\Av[Z_i(\Be)Z_j(\Be)]$, we can use the expansion
  properties of $G$, which imply that as $i,j$ grow apart,
  $Z_i(\Be)$ and $Z_j(\Be)$ become more and more independent, and so
  $\Av[Z_i(\Be)Z_j(\Be)] \to \Av[Z_i(\Be)]\cdot \Av[Z_j(\Be)] =
  \left(\frac{|F|}{|E|}\right)^2$. The exact statement is that for
  $i>j$,
  \begin{align}
    \Av[Z_i(\Be)Z_j(\Be)] \le \frac{|F|}{|E|}
      \left( \frac{|F|}{|E|} + |\lambda|^{i-j-1}\right) \ .
  \end{align}
  The proof of this fact is standard, and is given in
  Ref.~\cite{ref:Din07}, and will therefore be omitted.
  
  Inserting this into \Eq{eq:Z2}, we arrive to 
  \begin{align}
    \Av[Z^2(\Be)] &= t\frac{|F|}{|E|} 
      + 2\frac{|F|}{|E|}\sum_{i=1}^t\sum_{j=i+1}^t
        \left( \frac{|F|}{|E|} + |\lambda|^{i-j-1}\right) \\
     &\le t\frac{|F|}{|E|} + t(t-1)\left(\frac{|F|}{|E|}\right)^2 +
      \frac{2t}{1-\lambda}\frac{|F|}{|E|} \\
     &= t\frac{|F|}{|E|}\left(1 + \frac{|F|}{|E|}(t-1) +
     \frac{2}{1-\lambda}\right) \ .
  \end{align}
  Using \Eq{eq:positive}, we get
  \begin{align}
    \UNSAT_\sigma(G^t) \ge  
      \frac{t\UNSAT_\sigma(G)}
        {1 + t\UNSAT_\sigma(G) + \frac{2}{1-\lambda}}
        \EqDef F\big(\UNSAT_\sigma(G)\big) \ ,
  \end{align}
  where $F(x) = \frac{tx}{1 + tx + \frac{2}{1-\lambda}}$.
  If $x\le 1/t$, $F(x) \ge \frac{tx}{2+\frac{2}{1-\lambda}}$. On the
  other hand, as $F(x)$ is monotonically increasing for $x>0$, then
  for $x\ge 1/t$, $F(x) \ge F(1/t) = \frac{1}{2+\frac{2}{1-\lambda}}$.
  Setting $x=\UNSAT_\sigma(G)$ and using Equation \ref{def:c}   
  completes the proof.
\end{proof}

\bibliographystyle{hep}
\bibliography{QC}

\end{document}